\documentclass{article}
\usepackage{amsthm}
\newtheorem{theorem}{Theorem}
\newtheorem{lemma}{Lemma}
\newtheorem{definition}{Definition}
\newtheorem{proposition}{Proposition}
\newtheorem{corollary}{Corollary}
\usepackage{amssymb}
\usepackage{tikz}
\usetikzlibrary{arrows.meta,positioning} 
\usepackage{arxiv}
\usepackage{xcolor}
\usepackage[utf8]{inputenc} 
\usepackage[T1]{fontenc}    
\usepackage{hyperref}       
\usepackage{url}            
\usepackage{booktabs}       
\usepackage{amsfonts}       
\usepackage{nicefrac}       
\usepackage{microtype}      
\usepackage{lipsum}
\usepackage{amsmath}
\usepackage{graphicx}
\graphicspath{ {./images/} }
\usepackage[style=authoryear,backend=biber]{biblatex}
\DeclareFieldFormat{urldate}{}

\AtEveryBibitem{\clearfield{month}}

\title{Evolution of Plastic Dispersal in Stable Environment: Local Information of Fitness Consequence}

\author{
 Wayne Liang \\
  Department of Zoology\\
  University of Cambridge\\
  Downing Street, Cambridge, CB2 3EJ UK \\
  \texttt{wl415@cam.ac.uk} \\
   \And
 Rufus Johnstone \\
  Department of Zoology\\
  University of Cambridge\\
  Downing Street, Cambridge, CB2 3EJ UK \\
  \texttt{raj1003@cam.ac.uk}
}

\begin{document}
\maketitle
\begin{abstract}
Fitness consequence of dispersal depends on property of the entire landscape, which patches are available and what are the cost of moving. These are information that are not available locally when an organism make the decision to disperse. This poses a problem to the organism, where it is unclear how an adaptive decision can be made. This also poses a problem to the scientist, since in order to study the adaptiveness of dispersal, we need information of the entire landscape. For theorist, this is through making a series of assumption about either the landscape or the organism, and for empiricists, this means a large amount of measurements needs to be made across a large area. In this paper, we propose a link between local demographic process, which an organism can have access to, to the fitness consequence of dispersal. This meant local environmental cue can be used for the decision on dispersal, and hence allow the evolution of plastic dispersal strategy. We will then show that using this approach, evolution of dispersal on complex landscape can be modelled with relative ease, and to show that accidental dispersal in one patch can drive the evolution of adaptive dispersal in another.   
\end{abstract}


\section{Introduction}
Dispersal is an important trait that interacts with various aspect of biological evolution. It determines the pattern of gene flow allowing or restricting local adaptation \parencite{haldane1956RelationDensityRegulation}, changes relatedness structure and hence effects social evolution \parencite{temple2006DispersalPhilopatryIntergroup,escoda2017UsingRelatednessNetworks}, and alters population dynamics, affecting the persistence of the population \parencite{schreiber2010InteractiveEffectsTemporal}. Due to its central importance, it has long been a subject of theoretical investigation, which have made clear that dispersal generates strategic feedbacks: an individual’s optimal dispersal decision both depends on, and alters, the dispersal environment experienced by others. For example, Hamilton and May have shown that the evolution of dispersal can be shaped by kin selection \parencite{hamilton1977DispersalStableHabitats}, while other models have shown that dispersal is disfavored in the absence of temporal heterogeneity even with spatial variation \parencite{hastings1983CanSpatialVariation}, but favored under changing environmental conditions \parencite{parvinen2020EvolutionDispersalSpatially,parvinenEvolutionDispersalSpatiotemporal2023}. Taken together, these studies show that the evolution of dispersal is inherently frequency-dependent, with the fitness consequences of individual dispersal decisions shaped by the dispersal behaviour of others and by environmental variability.. 

For a decision on dispersal to be adaptive, it must on average provide fitness benefit to the decision maker. However, since dispersal is a choice between staying locally or moving to somewhere else in space, for the decision to be reliable it requires the focal individual to have information about the global environment. A similar constraint works for scientists as well: for a theorist, to solve for the optimal dispersal pattern requires global information about the metapopulation. Consequently, in order for models to remain tractable, modellers must often make strong assumptions of global homogeneity or symmetry. And for a empiricist, it is necessary to monitor across a large landscape to discern the fitness consequence of dispersal.

To deal with  these issues, theorists have mainly resorted to two approaches. Firstly, to let the disperser be omniscient \parencite{fretwell1969TerritorialBehaviorOther}, or at least allow them to sample patches with very small costs \parencite{ollason2001ApproximatelyIdealMore,forsman2022EvolutionSearchingEffort}, these are the ideal-free-distribution models. On the other hand, another approach is to let the environment be constant over an evolutionary timescale \parencite[e.g.][]{taylor1988InclusiveFitnessModel}, or at least to assume that the pattern of environmental change is constant \parencite[e.g.][]{cohen1991DispersalPatchyEnvironments}. This renders the environment  predictable, so that selection favours genotypes that act "as if" they are aware of the pattern of change. Whilst these methods have been applied with great success, assumption of low-cost sampling and environmental stability over evolutionary timescales limits the scope of natural systems to which the theory can be applied to. Furthermore,  the need for simplifying assumptions of global symmetry means that most theoretical models must work with extremely simplified landscapes.

In this chapter, I aim to develop a model of dispersal that requires only stability in the \textit{demographic }timescale, which could well be within a single generation and does not require the stability over evolutionary timescale as is often presumed. I will show that there is sufficient information in the local demographic processes to determine the fitness consequence of dispersal. The benefit is three folds: Firstly and most importantly, this shows there are sufficient environmental cue locally for a plastic dispersal strategy to evolve, allowing different probability of dispersal in different patches following the same reaction norm. Secondly, this is a new tool for modelling. Since optimal dispersal strategies can be determined by local population attributes, models can be greatly simplified, and consequently, much more realistic assumptions about the landscape and the mode of dispersal can be made. Lastly, if the theoretical result is validated empirically, cost of measuring the fitness consequence of dispersal would be greatly reduced, since one need only measure rate of movement into and out of the focal patch, without the need to track individuals across the landscape. 

We would first illustrate the approach through a simplified model. First, we will show that the fitness consequence of dispersal from a patch can be directly inferred from the local rate of immigration and emigration. Secondly, we will show how individuals can detect the fitness consequence of dispersal through knowledge of local population growth rate. In the third section, we will show the power of this modelling framework, through modelling evolution of dispersal in complex landscape. We will also show how accidental dispersal can drive the evolution of adaptive dispersal, even in the absent of kin selection and temporal heterogeneity, in contrast with previous results \parencite[]{hastings1983CanSpatialVariation, parvinenEvolutionDispersalSpatiotemporal2023}. Lastly, we will show how kin selection can be incorporated into this modelling framework.

\section{Models}
\label{sec:models}
\subsection{Taster Model}
\label{subsec: taster model}
To illustrate the general approach, an island model is built where the within-patch dynamic is the same as those in \textcite{taylor1988InclusiveFitnessModel} but the assumptions about global parameters are relaxed. Specifically, we assume that for a focal patch, there are $N$ breeding site, and for each adult who occupies a breeding site, they produce $c$ offspring. Some of these offspring then disperse with probability $d$ resulting in $Ncd = \mathbf{o}$ emigrants, and the patch receive a total number of $\mathbf{i}$ juvenile immigrants. All the adults then dies, with the juveniles competing for the breeding site equally. Those juveniles that does not win a breeding site dies. The only assumption about the global property is that the population dynamics has reached an steady state so that reproductive value of each class is well defined. 

\subsection{General Model}
\label{subsec:general_model}
Suppose that there are a number of $N$ patches, each with their internal dynamics. Suppose that given a plastic dispersal strategy, it is possible to write the population dynamic as:
\begin{align}
    \mathbf{n}_{t+1} = M \mathbf{n}_t
\end{align}
where the $i^{th}$ index of $\mathbf{n}_t$ is the population size of patch $i$ at time $t$. Suppose that $M$ is both ergodic and positive semidefinite. The diagonal entries of the matrix $M$ thus represents the internal dynamic of each patches, and the off-diagonal entries $M_{ij}$ represent the movement from patch $j$ to $i$. Since we only wish to study the decision to disperse, let's suppose that the probability of an disperser to arrive in patch $i$ provided they are from patch $j$ is fixed regardless of dispersal probabilities in either patches, provided they have dispersed. We can then rewrite the matrix $M$ as:
\begin{align}
    M = (I-D)S + DT
\end{align}
Where $D$ is a diagonal matrix with indices $D_{ii}$ being the dispersal probability at patch $i$, $S$ a diagonal matrix for the within-patch dynamic that an organism would experience in case they did not disperse, and $T$ a matrix where the entry $T_{ij}$ encodes the probability that an dispersed individual who left patch $j$ enters patch $i$. 

\subsection{Model for Non-Linear Local Population Growth Rate}
\label{subsec:non_linear}
Suppose on a meta population, the dynamic of each patch can be described by the following equation:
\begin{align}
    \Delta n_{i} = b_i(n_{i,t}) - d_i(n_{i,t}) +\mathbf{i}-\mathbf{o} \label{eq:discrete non linear}
\end{align}
in discrete time, or
\begin{align}
    \frac{dn_i}{dt} = b_i(n_{i,t}) - d_i(n_{i,t}) +\mathbf{i}-\mathbf{o}
\end{align}
 in continuous time, where $b_i$ and $d_i$ is the patch specific function that gives us birth and death rate respectively on the patch given population size, and $\mathbf{o}$ and $\mathbf{i}$ denotes the efflux and influx to the patch via dispersal.

\section{Results}
\subsection{Direct Fitness Benefit of Dispersing can be Deduced from the Net Flux into the Patch}
We show that the direct fitness consequence of dispersing, as opposed to staying, is equal to the netflux into the patch. Or:
\begin{align}
    \Delta w = \mathbf{i} - \mathbf{o} \label{eq:fitness consequence of dispersal}
\end{align}
This can be seen by analysis of the taster model (subsection \ref{subsec: taster model}), the result also holds more generally, allowing for overlapping generation and growing or shrinking population sizes. This can be seen through the analysis of the general model (subsection \ref{subsec:general_model}) which will be reserved in the appendix \ref{app:general_model}. Since in the taster model, the population size does not change, the reproductive value of an adult at time $t$ is the same as an adult in time $t+1$, or:
\begin{align}
    V_{a, t} = V_{a,t+1} = V_a
\end{align}
where $V$ denotes the reproductive value, and the subscript $a$ denoting the adults. Since the adult dies in that time step, it's reproductive value is equal to the sum of the reproductive value of all offspring. We thus have:
\begin{align}
    V_{a} = cdV_d + c(1-d)V_s 
    \label{eq:repro_val_adult}
\end{align}
where $V_d$ and $V_s$ denotes the reproductive value of juveniles that dispersed and stayed respectively, $c$ the fecundity of adult and $d$ the probability of a juvenile to disperse. For a staying young, it gains a reproductive value of $0$ if it failed to win a breeding spot in the competition, and a reproductive value of $V_a$ if it succeed. Thus the reproductive value of a staying young is simply the probability it gain a breeding spot times the reproductive value of an adult. In equation terms:
\begin{align}
    V_s = \frac{N}{Nc(1-d) + \mathbf{i}}V_a
\end{align}
or equivalently:
\begin{align}
    V_a = \frac{Nc(1-d) + \mathbf{i}}{N}V_s
    \label{eq:repro_val_young}
\end{align}
Substituting equation \ref{eq:repro_val_young} into equation \ref{eq:repro_val_adult}, we get the expression:
\begin{align}
    (Nc(1-d)+\mathbf{i})V_s = NcdV_d + Nc(1-d)V_s
\end{align}
Cancelling out the terms and noting that $Nc$ is the total number of locally born juvenile, and $d$ is the dispersal probability, therefore $Ncd$ is the total number of emigrant from the patch, we have:
\begin{align}
    \mathbf{i}V_s = \mathbf{o}V_d
\end{align}
or equivalently:
\begin{align}
    \frac{V_d}{V_s} = \frac{\mathbf{o}}{\mathbf{i}}
\end{align}
Since both reproductive value and the flux of individual are positive, and that reproductive value can be scaled arbitrarily as long as the ratio are conserved, we can deduce that:
\begin{align}
\label{eq:substitution1}
    V_d&=\mathbf{i}\\
    V_s&=\mathbf{o} \label{eq:substitution2}
\end{align}
For an individual to leave the patch, effectively, it reduces the number of "offspring" that stayed by one, and increased the number of "offspring" that left by one. Therefore, the fitness consequence is:
\begin{align}
    \Delta w = V_d-V_s = \mathbf{i}-\mathbf{o}
\end{align}
Which equals to the net flux of movements into the patch.

\subsection{Population Growth Rate Equals Migration Flux in a Stable Population}

We will demonstrate in this section that the population growth rate is the same as the net flux out of the patch in a stable population. Suppose that we have a meta-population as per section \ref{subsec:non_linear} where within each patch, the dynamic is described by equation \ref{eq:discrete non linear}. Then, suppose that the population reached a fixed stable state such that the population size is constant at each patch. We then have:
\begin{align}
    \Delta n_i =  b_i(n_{i,t}) - d_i(n_{i,t}) +\mathbf{i}-\mathbf{o} = 0
\end{align}
Therefore:
\begin{align}
     b_i(n_{i,t}) - d_i(n_{i,t}) = \mathbf{o}-\mathbf{i}
\end{align}
Since at equilibrium, the dynamics can be linearized to make a matrix population model, through derivation in \ref{app:general_model}, we know that the fitness of dispersal is as described in equation \ref{eq:fitness consequence of dispersal}, therefore we have:
\begin{align}
    b_i(n_{i,t}) - d_i(n_{i,t}) = -\Delta w
\end{align}
Since only the sign of $\Delta w$ matters, we can divide the local population growth rate on the left hand side by the local population size, obtaining the per capita growth rate. 
Therefore, the per capita growth rate without dispersal provides sufficient environmental cue for an organism to determine the direct fitness consequence of dispersal, thereby allowing evolution of plastic dispersal strategy. 
\subsection{Evolutionary Stable State of Dispersal Matches the Influx to Efflux}
We will show that the evolutionary stable state (ESS) of plastic dispersal strategy when kin selection is ignored is to match efflux to influx locally. The local population growth rate is the sufficient environmental cue that an organism can use to achieve this. I will then show that this is unattainable provided there are dispersal cost, recovering canonical results as found in \textcite{hastings1983CanSpatialVariation}. However, I will also show that provided there are accidental dispersal, adaptive dispersal can evolve, and there exist a way to solve for the ESS in arbitrary cases. 

\subsubsection{Adaptive Dispersal cannot Evolve with Homogeneous Environment without Kin Selection}
\label{subsec:nodisperse}
Since the fitness consequence of dispersal instead of staying is equal to the net-flux into the patch, and at ESS, different strategy should have the same fitness consequence \parencite[]{bishopGeneralizedWarAttrition1978}, thus ESS is the case where net flux is zero in each of the patches. Or:
\begin{align}
    \Delta w_i = 0 \Rightarrow \mathbf{i}_i = \mathbf{o}_i
\end{align}
Where the subscript $i$ denotes the fitness consequence of dispersing relative to staying in patch $i$, and the number of emigrant and immigrant in patch $i$. Specifically, since decision to disperse locally can only change the number of emigrant but not immigrant, ESS is attained through matching the number of emigrant locally to the given number of immigrant. However, since all immigrant were emigrant, there are always less immigrant than emigrant across the landscape provided there is a cost of dispersal. Or:
\begin{align}
    \sum_i\mathbf{i}_i < \sum_i\mathbf{o}_i
\end{align}
Therefore, on average, natural selection selects for reduced dispersal propensity across all patches. This is a generalized result of \textcite{hastings1983CanSpatialVariation} where it is noted even with spatial variation, dispersal is selected against in the absence of kin selection and temporal variation.

\subsubsection{Adaptive Dispersal can be Driven by Accidental Dispersal}
\label{subsec:accidental_disp}
However, if accidental dispersal, which is prevalent in nature \parencite[]{burgess2016WhenDispersalDispersal}, is incorporated into the model, then adaptive dispersal can evolve. This is because, from the above result, we know that there is always going to be some patches where there are more emigrant then immigrant, therefore dispersal is selected against. However, if emigrant from these patches cannot be further reduced due to accidental dispersal, the evolutionary dynamic would be stable in these patches. These patches, now with a stable stream of emigrant, can then supplement the lost disperser such that in all other patches the number of immigrant and emigrant is matched.

Formally, suppose at each patch there is a minimum and maximum dispersal probability $p_{min,i}$ and $p_{max,i}$ respectively. The number of individuals in each patch being $n_i$, then the minimum and maximum number of emigrant it can produce is $\mathbf{o}_{min,i} = p_{min,i}n_i$ and $\mathbf{o}_{max,i} = p_{max,i}n_i$ respectively. Suppose the probability of an disperser leaving patch $j$ to successfully arrive at patch $i$ is fixed at $T_{ij}$, then we have:
\begin{align}
    \mathbf{i}_i = \sum_jT_{ij}\mathbf{o}_j
\end{align}
Or equivalently:
\begin{align}
    \hat{\mathbf{i}} = T\hat{\mathbf{o}}
\end{align}
where $\hat{\mathbf{i}}$ and $\hat{\mathbf{o}}$ are the vectors whose $i^{th}$ index is the number of immigrant and emigrant at patch $i$ respectively, and $T$ the dispersal matrix whose entry $T_{ij}$ is the probability of an disperser from patch $j$ to arrive at patch $i$. The ESS across the landscape is thus the solution of:
\begin{align}
\label{eq:ESS}
    p_i:\begin{cases}
        p_i = p_{max,i}, \hat{\mathbf{i}}_i \geq \hat{\mathbf{o}}_i\\
        p_{min,i}<p_i<p_{max,i}, \hat{\mathbf{i}}_i = \hat{\mathbf{o}}_i
        \\
        p_i = p_{min,i}, \hat{\mathbf{i}}_i \leq \hat{\mathbf{o}}_i
    \end{cases}
\end{align}
and 
\begin{align}
    \hat{\mathbf{o}} = \hat{p}\odot\hat{n}
\end{align}
where $\hat{p}$ is a vector of dispersal strategy at each patch, $\hat{n}$ vector of population size, and $\odot$ denotes element-wise multiplication. The equations in \ref{eq:ESS} means at ESS, the maximum dispersal probability is stable if there are more immigrant than emigrant despite the maximum being reached, and minimum dispersal probability is stable if emigrant under minimal dispersal still outnumbered immigrant. In all other patches the number of emigrant would match the number of immigrant. 

It can be proven that given a transition matrix $T$, a population size vector $\hat{n}$, the maximum and minimal dispersal probability $\hat{p}_{min}$ and $\hat{p}_{max}$, there exist a unique ESS (See appendix \ref{Appendix: proof}) and it can be found with an algorithm (See appendix \ref{Appendix: algo}). It can also be shown that when dispersal cost is low, the efflux vector $\hat{\mathbf{o}}$ at ESS is approximately the leading eigenvector of the transition matrix $T$, providing it is a permissible strategy. This means if we consider a landscape network with each mode representing a patch and edges the probability of movement, then ESS at low dispersal cost would lead to efflux proportional to the eigen vector centrality of each patch (See appendix \ref{app:eigen}). 

To test the prediction, we compare, in a simple two patch model, the analytical result and simulated evolution with plastic dispersal strategy. See appendix \ref{Appendix: IBM} for details on the simulated evolution. We imagine a case where there are two patches, each hosting 500 adults. On the "small" island, each adult produce 10 offspring, whilst on the "large" island, each adult produce 50 offspring. On these island, the organisms follows life cycled as described in the taster model (section \ref{subsec: taster model}). We let the minimum dispersal probability to be 0.2 and maximum 1 across both island, and varied the dispersal cost, which is the probability of an dispersed individual failed to arrive at the other patch. By solving equation \ref{eq:ESS}, we find that the ESS for dispersal on the "large" island is always $p_{min}$, and for dispersal cost less than 0.8, efflux of the "small" island matches the influx, leading to the dispersal probability of:
\begin{align}
    d_{small} = (1-k)\frac{n_{large}}{n_{small}}p_{min}
\end{align}
Where k is the dispersal cost, $n_{large}$ and $n_{small}$ the number of juvenile in the large and small patch before the dispersal phase, being 25000 and 5000 respectively, and $p_{min}$ being 0.2. Substituting in the values, the evolutionary stable strategy for small patch equals to $1-k$ in this particular setting. When the dispersal cost is greater than 0.8 the ESS for both patches are $p_{min}$ since influx into both patches are less than efflux by accidental dispersal alone. Indeed, we see that the individual based model confirms our analytical prediction (Fig. \ref{fig:ibm}). 

\begin{figure}[h]
    \centering
    \includegraphics[width=0.7\linewidth]{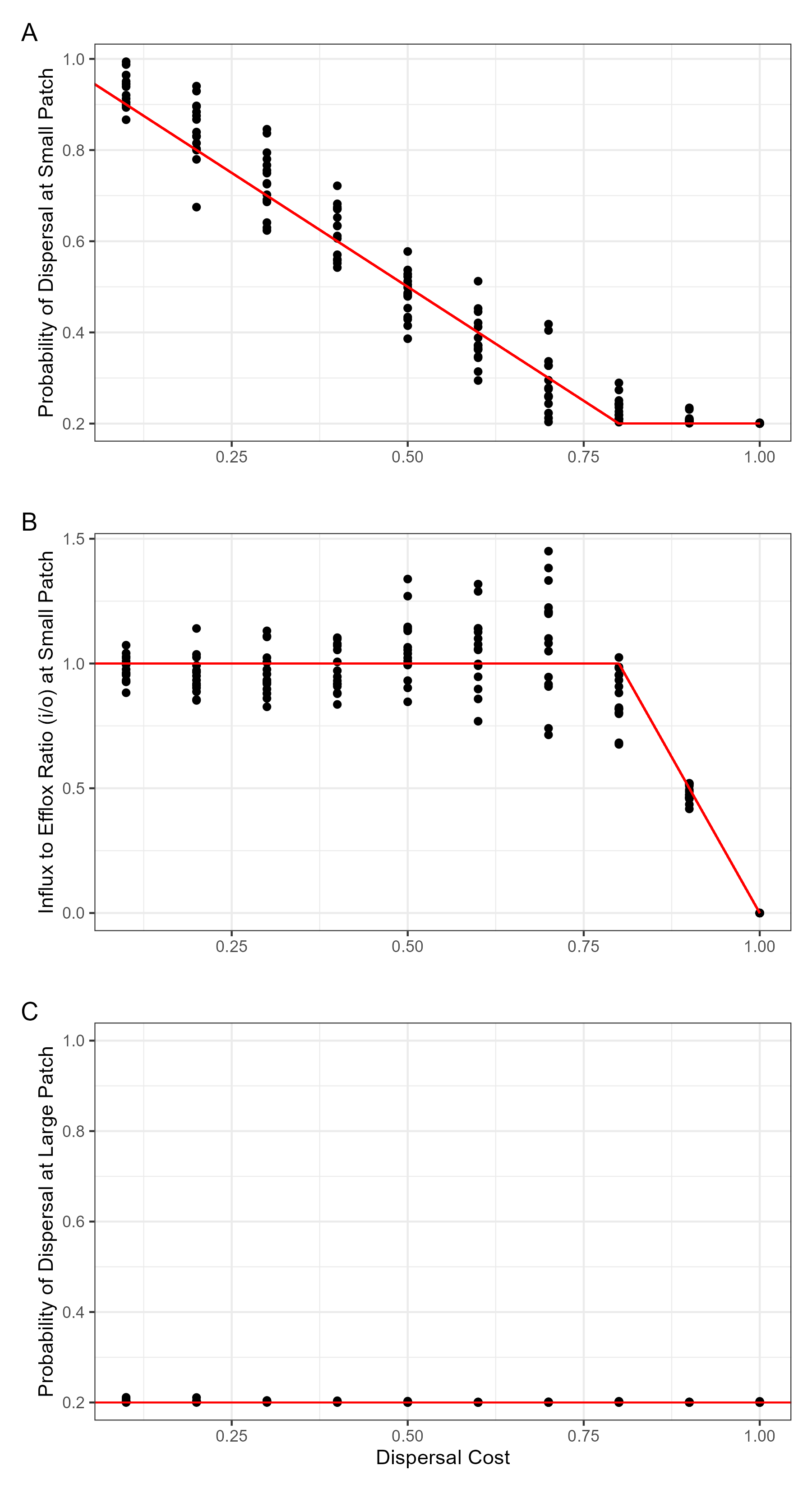}
    \caption{Box plots of 15 realisations of two island dispersal scenario in simulated evolution. Red line is the prediction from the analytical result. the x axis is dispersal cost with y-axis being the A) dispersiveness of the smaller patch, B) the ratio between immigrant and emigrant in the small patch, and C) the dispersiveness in the large patch. }
    \label{fig:ibm}
\end{figure}
\subsection{Incoorporating Kin Selection}
\label{subsec:kin}
Since the work of \textcite{hamilton1977DispersalStableHabitats}, we know that kin selection is important in shaping dispersal strategy. In this section we well analyze the taster model (section \ref{subsec: taster model}) to gain an intuition. The analysis of the general model (section \ref{subsec:general_model} with kin selection using  Generalizing the approach from \textcite{taylor1988InclusiveFitnessModel} to analyze the taster model (section \ref{subsec: taster model}), we have the expression:
\begin{align}
    \Delta w = V_d - V_s + V_sr_s -V_dr_d
\end{align}
Where $V_d$ and $V_s$ is the reproductive value of dispersing and staying respectively, $r_d$ and $r_s$ is the relatedness to those they compete with if they disperse or stay respectively. This is because, when an organism disperse, it displace a young where it disperse to, with an average reproductive value of $V_d$, which result in a inclusive fitness consequence of $V_d r_d$, but reduce competition in the natal patch such that the combined offspring pool gain the advantage $V_s$, leading to the focal individual gaining inclusive fitness of $V_sr_s$. Substituting in the relation between reproductive value and migration rate (equation \ref{eq:substitution1}, \ref{eq:substitution2}) we get the expression:
\begin{align}
    \Delta w = \mathbf{i} - \mathbf{o} + r_s\mathbf{o}-r_d\mathbf{i}
\end{align}
Since $\Delta w$ must be zero at ESS, we will have 
\begin{align}
    \frac{\mathbf{i}}{\mathbf{o}} = \frac{1-r_s}{1-r_d}
\end{align}
If we further assume isolation by distance, therefore $r_s \geq r_d$, we have:
\begin{align}
   \mathbf{i} -\mathbf{o} + r_s\mathbf{o} \geq \Delta w \geq (\mathbf{i} - \mathbf{o})(1-r_s)
   \label{eq:bounds_of_fitness}
\end{align}
Importantly, all parameters in equation \ref{eq:bounds_of_fitness} are present locally. This means the organism can, in principle, gain information locally to estimate the inclusive fitness of dispersal decision.

\section{Discussion}
\subsection{So, Why does This Work?}
In this paper, we show that the fitness consequence of dispersal in a stable environment is simply the net flux into the patch. At first this might seems counter intuitive, as the fitness consequence of dispersal depends on factors across the landscape. An intuitive way to think about this is as follows. Since in the population equilibrium, population size in each patch remains constant. Thus, the fitness of staying in the local patch depends on the extent of competition. If there are more influx than there are efflux in a patch, than individuals in that patch experience more competition than the average meta population and vice versa. Similarly, when the metapopulation declines or increases at a constant rate, with the ratio between population sizes in different patches held constant, the same logic applies. Whilst this intuitive idea may aid understanding, it can lead to paradoxes. For example, even if all the emigrant from a focal patch ends up in patches with worse competition, provided the focal patch has more immigrant than emigrant, it still pays for an organism to disperse (Fig. \ref{fig:paradox}).

\begin{figure}[ht]

\centering

\begin{tikzpicture}[
  font=\small,
  >=Latex,
  pop/.style={draw, rounded corners, align=center, minimum height=12mm},
  source/.style={pop, minimum width=38mm, minimum height=18mm},
  smallpop/.style={pop, minimum width=30mm},
  flux/.style={-Latex, line cap=round},
  veryhigh/.style={flux, line width=2.6pt},
  high/.style={flux, line width=2.0pt},
  med/.style={flux, line width=1.4pt},
  low/.style={flux, line width=0.9pt}
]

\node[source] (S) {Source\\(large population)};
\node[smallpop, below left=22mm and 36mm of S] (P1) {Small pop 1\\\textbf{(focal)}};
\node[smallpop, below right=22mm and 36mm of S] (P2) {Small pop 2};

\draw[med]
  (S.south west)
  .. controls +(235:18mm) and +(70:18mm)
  .. node[midway, above left] {\textbf{0.20}}
  (P1.north);

\draw[veryhigh]
  (S.south east)
  .. controls +(305:18mm) and +(110:18mm)
  .. node[midway, above right] {\textbf{0.35}}
  (P2.north);

\draw[high]
  (P1.east)
  .. controls +(10:18mm) and +(170:18mm)
  .. node[midway, above] {\textbf{0.10}}
  (P2.west);

\draw[low]
  (P2.west)
  .. controls +(190:16mm) and +(350:16mm)
  .. node[midway, below] {\textbf{0.05}}
  (P1.east);

\draw[low]
  (P2.north)
  .. controls +(90:18mm) and +(270:18mm)
  .. node[midway, below] {\textbf{0.05}}
  (S.south);

\node[below=2mm of P1] {$\mathbf{i}_1 - \mathbf{o}_1 = 0.15$};
\node[below=2mm of P2] {$\mathbf{i}_2 - \mathbf{o}_2 = 0.35$};

\end{tikzpicture}
\caption{Migration structure with "paradoxical" fitness benefit. Arrow width and labels denote migration rates. Both small populations have positive net immigration, with $\mathbf{i}_1 - \mathbf{o}_1<\mathbf{i}_2 - \mathbf{o}_2$. The unit of migration being the small population size, which is assumed to be the same in the two small populations. For the ease of graphing dispersal cost is assumed to be 0. Even though staying in the focal population leads to less competition then leaving for small population 2, dispersal still brings fitness advantage according to the calculations. }
\label{fig:paradox}
\end{figure}
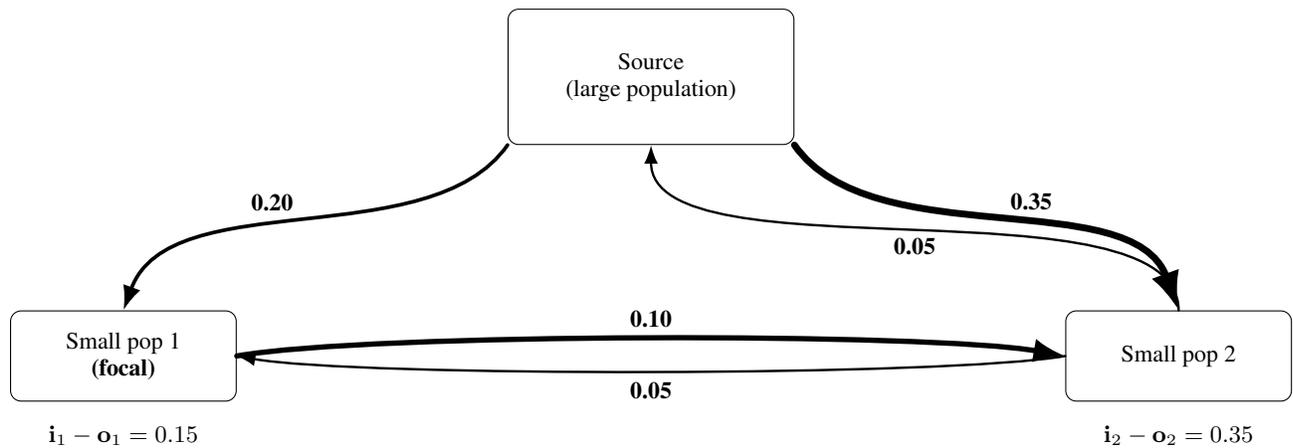

\subsection{Plastic Dispersal Strategy}
\label{subsec:plastic}
Plastic dispersal strategy is prevalent in nature (see \textcite{thierry2024InterplayAbioticBiotic} for a list of examples), and it has received much theoretical attention \parencite[]{fretwell1969TerritorialBehaviorOther, ollason2001ApproximatelyIdealMore, parvinen2020EvolutionDispersalSpatially}. An organism must have access to information of the fitness consequence of dispersal, conditioned on the current landscape, in order to make an adaptive decision. However, the fitness consequence of dispersal depends on conditions across the landscape. This meant that modeller need to build in mechansims through which organisms can gain reliable information of the fitness consequence of dispersal. This has traditionally been done in two ways: Firstly, organisms are assumed to have access to information of external environment \parencite[]{fretwell1969TerritorialBehaviorOther} or at least allowed to sample the landscape with some cost \parencite[]{ollason2001ApproximatelyIdealMore} with imperfect information \parencite[]{matsumura2010ForagingSpatiallyDistributed}. This have generated a wealth of literature with good empirical support \parencite[]{hache2013ExperimentalEvidenceIdeal, martin2022LargeHerbivoresPartially}. Whilst the assumption may be realistic for some organism where the cost of movement, and therefore sampling, are low, this framework clearly does not captures the whole range of biological possibilities. Secondly, some models posit that environment is stable over evolutionary time \parencite[]{hastings1983CanSpatialVariation} or at least the pattern of change is stable \parencite[]{cohen1991DispersalPatchyEnvironments,massolEvolutionDispersalSpatially2015,parvinen2020EvolutionDispersalSpatially,parvinenEvolutionDispersalSpatiotemporal2023}. The constancy of change meant that information on local patch quality is sufficient to deduce the fitness consequence of dispersal, as such, organism needs only sense the type of patch they are in. This meant that the genotype that "act as if they know" global information will be selected for. However, this modelling approach results in organisms with genetically determined dispersal probability conditioned on local patch quality. If the pattern of change changes, these strategies then become maladapted. Whilst this theoretical tradition had provided us with a wealth of intuition, the assumption of stable patterns of change over evolutionary timescale seems unrealistic, and these predictions are seldom tested empirically. 

In our model, we have shown that if we ignore kin selection, the fitness consequence of dispersal is strictly equal to the net migration into the patch in a stable environment, provided that the demography have stabilized. Demography is said to be stabilized when the ratio between the population size in each patch is stabilized. Crucially, this meant that there are sufficient information locally in the rates of movements. Hence organism can estimate the fitness consequence of dispersal without the need to leave the patch. Whilst it seems unclear how an organism will be able to estimate the net migration into and out of the patch, we have also shown that if the population size is stable over time, then an organism should stay when the rate of population growth is positive, and to disperse when it is negative. Provided, then, if an organism can can sense the per capita growth rate through it's own struggle or welfare, organisms can effectively make adaptive decision to dispersal. This is in line with our biological intuition and the "good-stay, bad-disperse" rule \parencite[]{hui2012FlexibleDispersalStrategies}. If the organism has the capacity to perform kin recognition, then, the inclusive fitness of dispersing can also be accurately estimated. Since this strategy is optimal in all stable environment, organism with dispersal strategy evolved in one environment will perform equally well in a different environment, provided the demography have time to stabilize. The same strategy also effectively collects information through local processes, widening the range of biological systems it applies to compared to the ideal-free distribution literature. 

\subsection{Accidental Dispersal Drives Evolution of Adaptive Dispersal in Stable Environment}

Previous theoretical work has shown that spatial variation alone is insufficient to favour the evolution of adaptive dispersal in stable environments when kin selection and temporal heterogeneity are absent \parencite{hastings1983CanSpatialVariation, holt1985PopulationDynamicsTwopatch, cohen1991DispersalPatchyEnvironments, gyllenberg2002EvolutionarySuicideEvolution, parvinen2006EvolutionDispersalStructured, parvinenEvolutionDispersalSpatiotemporal2023}. We recover this result in a general setting (Section \ref{subsec:nodisperse}): when dispersal can be reduced arbitrarily close to zero in all patches, the evolutionary stable state corresponds to minimal dispersal everywhere, because the total number of immigrants across the landscape is strictly smaller than the total number of emigrants when dispersal is costly. Under these conditions, selection consistently favours reduced dispersal, even in heterogeneous landscapes.

However, this conclusion hinges on a strong assumption about the set of viable strategies. Some past work explicitly set the viable dispersal strategy between 0 and 1 \parencite[ and our section \ref{subsec:nodisperse}]{cohen1991DispersalPatchyEnvironments,gyllenberg2002EvolutionarySuicideEvolution,parvinenEvolutionDispersalSpatiotemporal2023} and others focus only on whether a mutation with reduced dispersal will be selected for provided it have already arose \parencite{hastings1983CanSpatialVariation}. Both effectively allows probability of dispersal to be arbitrarily small. When this assumption is relaxed to allow for accidental dispersal alongside plastic dispersal strategy (Section \ref{subsec:plastic}), adaptive dispersal can evolve even in the absence of kin selection or temporal variation. Accidental dispersal imposes a lower bound on emigration from some patches, forcing a non-reducible outflow at a cost. This enforced emigration reduces competition locally in the source patch while increasing competitive pressure elsewhere. As a result, patches receiving accidental dispersers experience excess immigration relative to emigration, creating conditions under which dispersal can be favoured as a means of escaping locally inflated competition. Importantly, this process does not require coordination across patches: when immigration exceeds emigration in a patch, selection favours increased dispersal, whereas when emigration exceeds immigration, dispersal is selected against thus maintaining the balance. In this way, forced dispersal in some patches generates suppliments the loss of immigrant due to migration cost in the rest of the metapopulation so the balance can be maintained.

The effect of accidental dispersal on the evolutionary dynamics of dispersal strategies is illustrated in Fig. \ref{fig:accidental_disp}. The figure shows the selection gradients in the space of dispersal probabilities in a two-patch system. In the absence of accidental dispersal (left panel), selection drives both patches toward minimal dispersal, reproducing the classic result. When a lower bound on dispersal is imposed in (right panel), whilst the dispersal probability are still selected to decrease across the parameter space, since it is bounded by the minimum dispersiveness due to accidental dispersal it supplies sufficient migration to the small patch such that adaptive dispersal evolve there. Because the selection gradients depend only on local migration fluxes, the same logic extends to arbitrary landscapes, where constrained patches act as sources that shape dispersal evolution elsewhere.

Accidental dispersal is prevalent in natural systems \parencite{burgess2016WhenDispersalDispersal}, arising through a variety of mechanisms including physical constraints \parencite[]{naman2016CausesConsequencesInvertebrate}, pleiotropy with other traits \parencite{strathmann2002EvolutionLocalRecruitment}, and dispersal as a by-product of foraging \parencite{vandyck2005DispersalBehaviourFragmented} or predator avoidance \parencite{leblond2016CaribouAvoidingWolves}. For example, selection for increased height under light competition in plants can inadvertently increase dispersal distances, particularly in wind-dispersed species. A striking animal example is provided by flightless insects on sub-Antarctic and Antarctic islands, where strong wind induced accidental dispersal is so costly that they have evolved flightlessness \parencite[]{leihy2020WindPlaysMajor}. In these systems, selection has favoured both morphological and behavioural traits that reduce dispersal \parencite{peckham1971notes}. Nevertheless, accidental dispersal still occurs: individuals of the Antarctic midge Belgica can be transported passively on wind-blown debris or by birds, despite being wingless \parencite{peckham1971notes}. These examples illustrate that strictly zero dispersal is often biologically unattainable, supporting the assumption of a non-zero lower bound on dispersal probability.

Taken together, our results show that accidental dispersal fundamentally alters the evolutionary consequences of spatial variation. By constraining the set of viable dispersal strategies, it generates persistent migration fluxes that reshape local selection gradients and allow adaptive dispersal to evolve in stable environments. Crucially, because the resulting selection pressures depend only on locally realised migration fluxes, this mechanism is compatible with dispersal strategies based on local information and extends naturally to complex landscapes.
\begin{figure}
    \centering
    \includegraphics[width=1\linewidth]{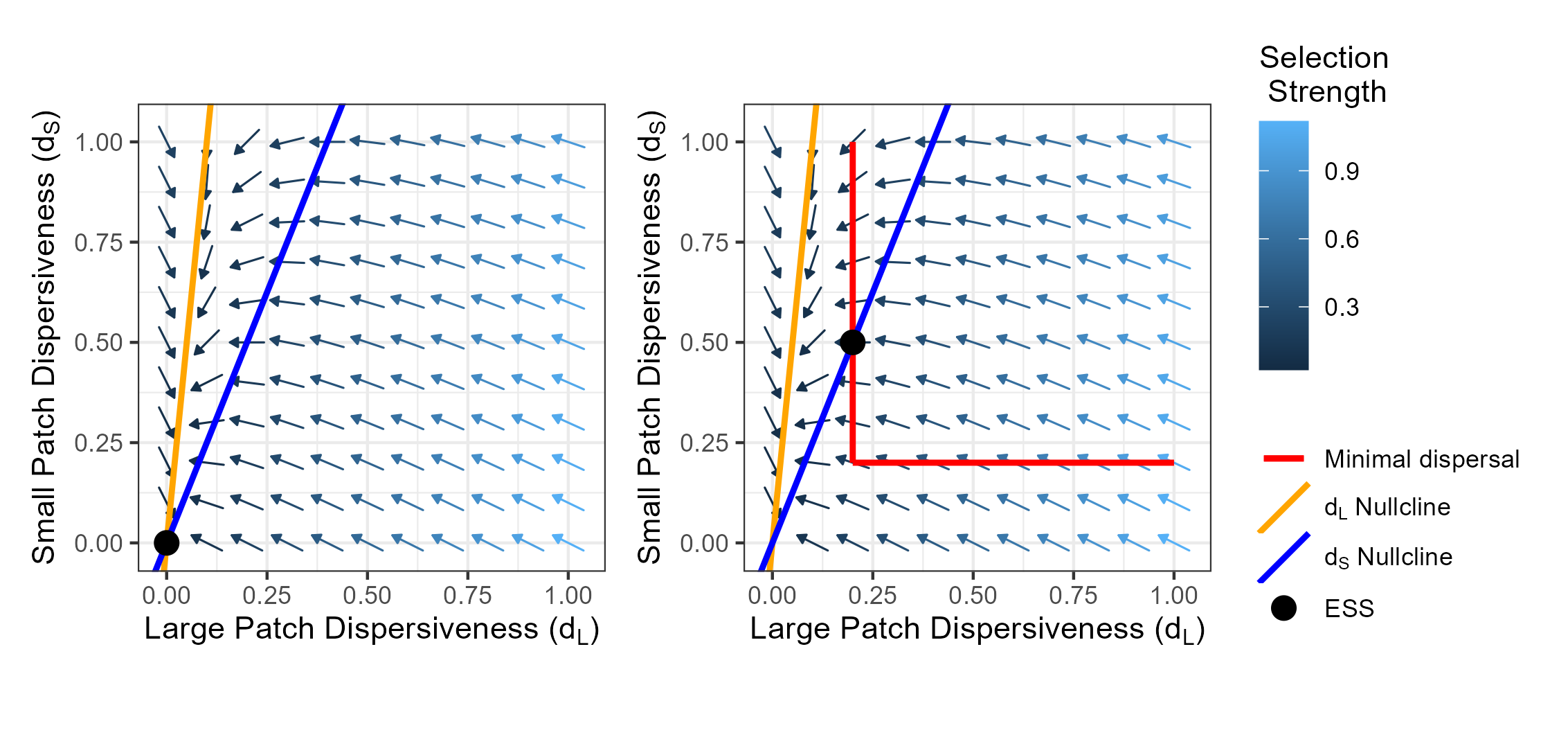}
    \caption{Selection gradients on dispersal probabilities in a two-patch system. Arrows indicate the direction and strength of selection. Left: when dispersal can be reduced to zero in both patches, selection drives both toward minimal dispersal. Right: imposing a lower bound on dispersiveness (red boundary) generates a persistent migration flux and changes the ESS where the large patch evolve to have minimal dispersal which leads to adaptive dispersal being selected for in the large patch. This demonstrates how accidental dispersal in one patch can drive the evolution of adaptive dispersal in another. The selection gradient is drawn by setting the large patch to be 5 times more productive than the small patch with dispersal cost of 0.5, and minimal accidental dispersal to be 0.2. }
    \label{fig:accidental_disp}
\end{figure}

\subsection{Theoretical and Empirical Applications}
\label{subsec:applications}

Although many theoretical studies of dispersal are, in principle, formulated for arbitrarily complex landscapes, practical analysis is often restricted to systems with only two or a small number of patch types in order to remain tractable \parencite{massolEvolutionDispersalSpatially2015,parvinen2020EvolutionDispersalSpatially, parvinenEvolutionDispersalSpatiotemporal2023}. This limitation reflects the difficulty of evaluating the fitness consequences of dispersal when those consequences depend on global landscape properties, including spatial structure and patch specific demography. Due to the difficulty in formulating analytical fitness expression in complex landscapes, studies on more realistic landscape are typically done with individual-based simulated evolution \parencite{tonkin2018RoleDispersalRiver, finand2024EvolutionDispersalMaintenance}.

The framework developed here substantially relaxes these constraints by reducing the problem of dispersal evolution to one governed by locally realised demographic processes. Because the direct fitness consequence of dispersal depends only on the balance of immigration and emigration at a focal patch, dispersal strategies can be analysed without explicit reference to global landscape property. If the landscape is represented as a directed network in which nodes correspond to patches and directed edges encode the probability that a dispersing individual from patch $j$ arrives in patch $i$, then the evolutionary stable dispersal strategy across the entire landscape can be characterised by solving a system of local balance conditions (Eq.~\ref{eq:ESS}). Given a dispersal transition matrix, patch-specific population sizes, and bounds on feasible dispersal probabilities, this system admits a unique evolutionary stable strategy (Appendix \ref{Appendix: proof}) that can be obtained algorithmically (Appendix \ref{Appendix: algo}), even in landscapes with arbitrary topology and asymmetric movement patterns.

This approach also facilitates the analysis of dispersal evolution in continuous landscapes. We illustrate this by deriving conditions under which adaptive dispersal evolves in a continuous landscape with spatial autocorrelation and organisms disperse via a dispersal kernel, with full derivations provided in the appendix \ref{app:conti}. We obtain the condition where adaptive dispersal is likely to evolve without the need of discretising the landscape.

The framework can be extended to incorporate kin selection in a straightforward manner. Kin selection plays an important role in the evolution of disprsal \parencite[]{hamilton1977DispersalStableHabitats, taylor1988InclusiveFitnessModel,irwin2001EvolutionAltruismSteppingStone,taylor2013SocialEvolutionDispersal}. Though the bulk of our paper considers cases where kin selection is ignored, provided that the process determining relatedness within and between patches is specified, kin selection enters into the fitness calculation without altering the central role of local migration fluxes (Section \ref{subsec:kin}). This allows kin-structured dispersal models to be analysed on complex landscapes without the need to track the global property of the landscape.

Beyond its theoretical advantages, the framework also has potential empirical applications. Estimating the fitness consequences of dispersal in natural systems is notoriously challenging, as it typically requires tracking individuals across large spatial scales and inferring fitness through indirect proxies such as survival or recruitment \parencite[]{doligez2008EstimatingFitnessConsequences}. A common complication is that individuals dispersing beyond the study area are often treated as having died, leading to systematic underestimation of dispersal benefits \parencite[]{doligez2008EstimatingFitnessConsequences,davis2023MultistateModelEstimate}. In contrast, our results predict that, in demographically stable systems, the direct fitness consequence of dispersal can be inferred from local net migration fluxes alone. If this prediction is validated empirically, it would substantially reduce the logistical and financial costs of studying dispersal by requiring only local measurements of immigration and emigration rates, rather than individual-based tracking across the landscape.


\newpage
\printbibliography[heading=bibliography]

\newpage
\appendix

\appendix
\section{Analysis of the General Model}
\label{app:general_model}
The general model found in section \ref{subsec:general_model} is detailed below. Recall that we define a matrix population model such that:
\begin{align}
    \mathbf{n}_{t+1} = M\mathbf{n}_t
\end{align}
where the $i^{th}$ index of $\mathbf{n}_t$ is the population size of patch $i$ at time $t$ and $M$ a transition matrix that defines the dynamic. In particular, we assume $M$ is ergodic and it can be decomposed into the following matrix:
\begin{align}
    M = (I-D)S+DT
\end{align}
Where $D$ is a dagonal matrix with indices $D_{ii}$ being the dispersal probability at patch $i$, $S$ a diagonal matrix with $S_{ii}$ denote the local population growth rate an non-dispersing individual would experience in a time step, and $T$ a matrix where the entry $T_{ij}$ denotes the probability that a dispersing individual in patch $j$ ended up in patch $i$. 

To calculate the reproductive value of an dispersing and an staying individual, we recall the fact that the reproductive value of patch $x$ is the $x^{th}$ index of the left leading eigen vector of the transition matrix \parencite[]{goodman1967ReconciliationMathematicalTheories, caswell1982StablePopulationStructure}. Therefore we have:
\begin{align}
    \lambda(D) v_x = (v^TM)_x =  (1-D_{xx})S_{xx} v_x + D_{xx}\sum v_iT_{ix}
\end{align}
Where $\lambda(D)$ is the leading eigen value of the transition matrix $M$ given a dispersal strategy $D$, $v_i$ the reproductive value of an individual at patch $i$, and the subscript $x$ denotes the focal patch. By defining the reproductive value of dispersing as:
\begin{align}
    v_{d,x} &= \sum v_iT_{ix}
\end{align}
we then have:
\begin{align}
    \lambda(D) v_x =   (1-D_{xx})S_{ii}v_x + D_{xx}v_{d,x}
\end{align}
Rearranging the terms leads to:
\begin{align}
    \frac{v_{d,x}}{v_x} = \frac{\lambda(D) - (1-D_{xx})S_{xx}}{D_{xx}}
\end{align}
Multiply the fraction on the right hand side by $\mathbf{n}_t/\mathbf{n}_t$ we get:
\begin{align}
    \frac{v_{d,x}}{v_x} = \frac{\lambda(D)\mathbf{n}_t - (1-D_{xx})S_{xx}\mathbf{n}_t}{D_{xx}\mathbf{n}_t}    
\end{align}
 If we assume that the population structure had reached the stable stage distribution, which is the right leading eigen vector, then, we have:
 \begin{align}
         \frac{v_{d,x}}{v_x} = \frac{\mathbf{n}_{t+1} - (1-D_{xx})S_{xx}\mathbf{n}_t}{D_{xx}\mathbf{n}_t}  
 \end{align}
 
 since $\lambda(D)\mathbf{n}_t = \mathbf{n}_{t+1}$. Observe that the denominator is simply the number of emigrant from the patch and the numerator the immigrant, we thus have:
 \begin{align}
      \frac{v_{d,x}}{v_x} = \frac{\mathbf{i}}{\mathbf{o}}
 \end{align}
 
To investigate the selection gradient on dispersal on a patch $x$, we can calculate the elasticity of the matrix $M$ to the patch specific dispersal rate $D_{xx}$ \parencite[]{caswell1982OptimalLifeHistories, taylorHowMakeKin1996}. We have the selection gradient:
\begin{align}
    \frac{dW}{dD_{xx}} = \sum_{i,j} v_i\frac{dM_{ij}}{dD_{xx}}u_j = u_x\sum_{i}v_i\frac{dM_{ix}}{dD_{xx}}|_{D_{xx} = D_{xx}^*}
\end{align}
where $v$ is the reproductive value or the left leading eigen vector of the matrix population model, $u$ the right leading eigen vector or the stable stage distribution, $dW/dD_{xx}$ the selection gradient on dispersal probabiltity in patch $x$ and $D_{xx}^*$ the mean population dispersal rate at patch $x$. The right hand side expression can be expanded to:
\begin{align}
     u_x\sum_{i}v_i\frac{dM_{ix}}{dD_{xx}} = u_x ((-1)v_x + \sum_{i\neq x}T_{ix}v_i) = u_x (v_{d,x} - v_x) = ku_x(\mathbf{o} - \mathbf{i}).
\end{align}
\section{Eigenvalue Centrality and Dispersal at Low Cost}
\label{app:eigen}
When there are no cost of dispersal, all strategies that leads to influx and efflux matching exactly are ESS, provided they are feasible. If we focus on only the second equation of \ref{eq:ESS}, then ESS would be the solution of the following equation:
\begin{align}
    \mathbf{i} = T\mathbf{o}
\end{align}
Since there are no dispersal cost, all columns sums to 1 and therefore the leading eigenvalue of $T$ is 1. Thus, all strategies that leads to the efflux in a patch proportional to the leading eigen vector of $T$ are evolutionary stable. Crucially, if we treat landscape as nodes and construct a directed network and the transition matrix $T$ the adjacency matrix of the landscape network, then the eigenvector is also the eigenvector centrality of the patch. Thus, when there are no dispersal cost, a solution to the ESS equation is that the efflux from each patch is the same as their eigenvector centrality. In the case where there are small dispersal cost, the ESS efflux would be slightly deviate from the solution at no dispersal cost. Except there will only be a single solution where a patch have minimal dispersal and the rest slightly smaller to the scaled eigenvector.

\section{Habitat network ESS and Algorithm}
\label{Appendix: algo}
The algorithm to find an ESS in a habitat network written in R is given below. It takes in variable networkk which is the $T$ matrix in the main text. $q$ the vector that denotes habitat quality, or, the maximum number of individual that can make a decision to disperse, and $p_{min}$ the minimal probability to disperse, which can either be a scalar (applied across all patches) or a vector.

\begin{verbatim}
    ESS <- function(networkk, q, pmins, maxit = 1000, 
    exact = TRUE, tol = 1e-8, printflux = FALSE)
{
  size = nrow(networkk)
  q <- c(q)
  pmins <- c(pmins)
  if(length(q) == 1)
  {
    q <- rep(q, size)
  }
  else if(length(q) != size)
  {
    print("error, q is of the wrong size")
  }
  if(length(pmins) == 1)
  {
    pmins <- rep(pmins, size)
  }
  else if(length(pmins) != size)
  {
    print("error, pmin is of the wrong size")
  }
  p <- pmins
  for(i in 1:maxit)
  {
    efflux <- q*p
    influx <- networkk%*%efflux
    net <- influx - efflux
    if(all(net > 0))
    {
      print(i)
      break
    }
    if(!exact)
    {
      check = all(near(net[p<1& p > pmins], 0, tol = tol))
    }
    else
    {
      check = all((net[p<1& p > pmins] == 0))
    }
    if(check&all(net[p == 1] > 0) & all(net[p == pmins] < 0))
    {
      if(printflux)
      {
        print(net)
      }
      return(p)
      break
    }
    valid_ind <- which(p<1 & influx > efflux)
    new <- influx[valid_ind]/q[valid_ind]
    p[valid_ind] <- pmin(1, influx[valid_ind]/q[valid_ind])
    
  }
  print("ESS finding exceeded run time")
  print(net)
  print(p)
  return(c("exceeded"))
}
\end{verbatim}

\section{IBM}
\label{Appendix: IBM}
The individual based model consists of two patches, A and B, with 500 individuals each. Every individual is characterized by two numbers, $z_a$ and $z_b$, denoting their dispersal trait from patch A and B respectively. At each timestep, all individuals reproduce in accordance with the fecundity of each patch, 10 offspring each on patch A and 50 on patch B. For each of the offspring and each of the traits,  they mutate with probability 0.05. If a mutation happens, the trait is updated to $z'_i = z_i + \epsilon$ where $\epsilon\sim\mathcal{N}(0,0.05)$. In the dispersal phase, each offspring $j$ in patch $i$ disperses with probability $p_j = 0.2 + 0.8/(1 + \exp(-z_{j,i}))$ such that the probability of dispersal is bounded between 0.2 and 1. Of those that disperse, they die with probability k, set differently in each simulation ranging from 0.1 to 1. Those that stayed do not incur any mortality. Of the dispersers that survived, they join the local offspring pool of the other patch. 500 individuals are then chosen in each offspring pool to replace the parents. The simulation started with populations of all $z$s drawn from a uniform distribution between [0,1]. This is ran for 30000 generations. For each value of k, simulation were ran 15 times. The codes are attached below.
\begin{verbatim}
import numpy as np
import numpy.random as rnd
import matplotlib.pyplot as plt
import time
from joblib import Parallel, delayed
import pandas as pd
rng = rnd.default_rng()
na = 500
nb = 500
feca = 10
fecb = 50
mutrate = 0.05
mindisp = 0.2
maxdisp = 1
times = 30000
inputs = []
for i in np.linspace(0.1,1,10):
    k = i
    for i in range(15):
        inputs.append((na,nb,feca,fecb,mutrate,mindisp,maxdisp,k,times))


def ibm(parms):
    na = parms[0]
    nb = parms[1]
    feca = parms[2]
    fecb = parms[3]
    mutrate = parms[4]
    mindisp = parms[5]
    maxdisp = parms[6]
    k = parms[7]
    times = parms[8]
    disprange = maxdisp - mindisp
    s = 1-k
    def link(x):
        #x[x < mindisp] = mindisp
        #x[x > maxdisp] = maxdisp
        #return x
        return disprange/(1+np.exp(-x)) + mindisp
    noffa = na*feca
    noffb = nb*fecb
    amean = np.zeros(times)
    bmeana = np.zeros(times)
    bmeanb = np.zeros(times)
    bmean = np.zeros(times)
    ia = np.zeros(times)
    ib = np.zeros(times)
    oa = np.zeros(times)
    ob = np.zeros(times)
    popa = rnd.random([na,2])
    popb = rnd.random([nb,2])
    st = time.time()
    for i in range(times):
        #at = time.perf_counter()
        #reproduction
        offspringa = np.tile(popa, (feca, 1))
        offspringb = np.tile(popb, (fecb, 1))
        #bt = time.perf_counter()
        #mutation
        #muta = rnd.binomial(1,mutrate,[noffa,2]).astype(bool)
        muta = rng.random(size = [noffa,2]) < mutrate
        #mutb = rnd.binomial(1,mutrate,[noffb,2]).astype(bool)
        mutb = rng.random(size = [noffb,2]) < mutrate
        offspringa[muta] += rnd.normal(0,0.05,np.sum(muta))
        #offspringa[offspringa < mindisp] = mindisp
        #offspringa[offspringa > maxdisp] = maxdisp 
        offspringb[mutb] += rnd.normal(0,0.05,np.sum(mutb))
        #offspringb[offspringb < mindisp] = mindisp
        #offspringb[offspringb > maxdisp] = maxdisp
        #dispersal
        #disperseda = rnd.binomial(1, link(offspringa[:,0])).astype(bool)
        disperseda = rng.random(size = offspringa.shape[0]) < link(offspringa[:,0])
        #dispersedb = rnd.binomial(1, link(offspringb[:,1])).astype(bool)
        dispersedb = rng.random(size = offspringb.shape[0]) < link(offspringb[:,1])
        stayeda = offspringa[np.logical_not(disperseda),:]
        stayedb = offspringb[np.logical_not(dispersedb),:]
        dispa = offspringa[disperseda,:]
        oa[i] = dispa.shape[0]
        dispb = offspringb[dispersedb,:]
        ob[i] = dispb.shape[0]
        #ct = time.perf_counter()
        #mortality
        dispasurvived = rng.random(size = dispa.shape[0]) < s
        dispbsurvived = rng.random(size = dispb.shape[0]) < s
        dispa = dispa[dispasurvived,:]
        ib[i] = dispa.shape[0]
        dispb = dispb[dispbsurvived,:]
        ia[i] = dispb.shape[0]
        #dt = time.perf_counter()
        #combining to form competing offspring
        compa = np.concatenate((dispb, stayeda), axis=0)
        compb = np.concatenate((dispa, stayedb), axis=0)
        #sampling a new generation
        popa = compa[rng.choice(compa.shape[0], na, replace=False)]
        popb = compb[rng.choice(compb.shape[0], nb, replace=False)]
        #calculating summary stats
        totalpop = np.concatenate((popa,popb), axis = 0)
        amean[i] = np.mean(link(totalpop[:,0]))
        bmean[i] = np.mean(link(totalpop[:,1]))
        bmeanb[i] = np.mean(link(popb[:,1]))
        bmeana[i] = np.mean(link(popa[:,1]))
        #et = time.perf_counter()
        #print(et-at,bt-at,ct-bt,dt-ct,et-dt)
    et = time.time()
    print(et - st)
    #a1 = np.zeros(na)
    atraitsa = link(popa[:,0])
    atraitsb = link(popa[:,1])
    btraitsa = link(popb[:,0])
    btraitsb = link(popb[:,1])
    #b1 = np.zeros(nb)
    #for i in range(na):
    #    a1[i] = np.mean(link(atraitsa[np.arange(na) != i]))
    #for i in range(nb):
    #    b1[i] = np.mean(link(btraitsb[np.arange(nb) != i]))

    return {"k": k, "oa": oa[-1], "ob": ob[-1], "ia":ia[-1],
    "ib": ib[-1],"ameana": np.mean(atraitsa),"bmeanb":np.mean(btraitsb),
    "amean":amean[-1],"bmean":bmean[-1],
    "parms":parms, "vara":np.var(link(totalpop[:,0])),
    "varb": np.var(link(totalpop[:,1])), 
    "cova": np.sum((np.mean(atraitsa) - amean[-1])*(atraitsa - amean[-1]))/(na-1),
    "covb": np.sum((np.mean(btraitsb)-bmean[-1])*(btraitsb-bmean[-1]))/(nb-1), 
    "covab": np.sum((atraitsa - amean[-1])*(np.mean(btraitsa)-amean[-1]))/na,
    "covba": np.sum((btraitsb - bmean[-1])*(np.mean(atraitsb) - bmean[-1]))/nb}


results = Parallel(n_jobs = 15)(delayed(ibm)(input) for input in inputs)
pd.DataFrame(results).to_csv("C:/Users/user/Documents/R/master/ibmofficial.csv")
\end{verbatim}
\section{Evolution of Adaptive Dispersal in Continuous Landscape with Dispersal Kernel}
\label{app:conti}
In this model we will aim to derive the condition under which we will expect adaptive dispersal to evolve in a continuous landscape without kin selection when organisms disperse with a dispersal kernel in order to illustrate the power of our approach. We will work on an arbitrary landscape and derive only when minimal dispersal is not evolutionary stable, without aiming to find an ESS.

Consider a continuous landscape with varying carrying capacity. Suppose dispersal happen at fixed discrete timestep $\Delta t$ during which individuals disperse with an arbitrary dispersal kernel on a two dimensional landscape. In between two dispersive steps, density regulation and reproduction occurs to yield a constant population density function with respect to location $f(\mathbf{x})$. At each time step, suppose a minimum proportion $p_{min}$ chose to disperse, the efflux is thus:
\begin{align}
    \mathbf{o} = p_{min}f(\mathbf{x^*})
\end{align}
The dispersal kernel $\Gamma(r)$ is the density distribution of distance traveled by a dispersing young. We thus have the density of young dispersing from point $\mathbf{x}$ to $\mathbf{x^*}$ being:
\begin{align}
    flux_{\mathbf{x}\rightarrow\mathbf{x^*}} = f(\mathbf{x})p_{min}(1-k)\frac{\Gamma(\|\mathbf{x} - \mathbf{x^*}\|)}{2\pi\|\mathbf{x} - \mathbf{x^*}\|}
\end{align}
where $\|\mathbf{x} - \mathbf{x^*}\|$ denotes the distance between the two points and $k$ being the dispersal cost incurred by all dispersing individual. The dispersal kernel is divided through by $2\pi r$ because of all the individuals dispersed a distance $r$ away, it is spread between $2\pi r$ localities. To sum up all the influx is thus:
\begin{align}
    \mathbf{i} = (1-k)p_{min}\int_\Omega\frac{f(\mathbf{x})\Gamma(\|\mathbf{x} - \mathbf{x^*}\|)}{2\pi\|\mathbf{x} - \mathbf{x^*}\|}d\mathbf{x}
\end{align}
Where $\Omega$ is the collection of all points. If we collect all points $\mathbf{x}$ a distance $r$ away from our focal point $\mathbf{x^*}$ into a set $\Omega_{\mathbf{x^*}}(r)$ we can rewrite the integral to:
\begin{align}
    \mathbf{i} = p_{min}(1-k)\int_0^\infty\frac{\Gamma(r)}{2\pi r}\int_{\Omega_{\mathbf{x^*}}(r)}f(\mathbf{x})d\mathbf{x}dr
\end{align}
Where if we define:
\begin{align}
    g_{\mathbf{x}^*}(r) =\frac{\int_{\Omega_{\mathbf{x^*}}(r)}f(\mathbf{x})d\mathbf{x}}{\int_{\Omega_{\mathbf{x^*}}(r)}\mathbf{1}\,\,d\mathbf{x}} = \frac{\int_{\Omega_{\mathbf{x^*}}(r)}f(\mathbf{x})d\mathbf{x}}{2 \pi r}    
\end{align}
where $g_{\mathbf{x}^*}(r)$ is the mean productivity of points $r$ away from $\mathbf{x}^*$, we get:
\begin{align}
     \mathbf{i} = p_{min}(1-k)\int_0^\infty\Gamma(r)g_{\mathbf{x}^*}(r) dr
\end{align}
Since dispersal yields a net direct fitness benefit if $\mathbf{i}>\mathbf{o}$, organisms are selected to disperse if:
\begin{align}
    (1-k)\int_0^\infty\Gamma(r)g_{\mathbf{x}^*}(r) dr > f(\mathbf{x}^*)
\end{align}
This simply states, that dispersal is selected for in localities where the productivity of an average neighbouring point times the dispersal cost, weighted by the dispersal kernel, is greater than the current patch. Some insights can be gained by further analyzing the equation. Since $g_{\mathbf{x}^*}(r)$ is the expected productiveness of a set of point, it must be true that: $\max(f(\mathbf{x})) \geq g_{\mathbf{x}^*}(r)\geq \min(f(\mathbf{x}))$ for all $\mathbf{x}^*$ and $r$. Hence we have:
\begin{align}
\int_0^\infty\Gamma(r)\max(f(\mathbf{x}))dr = \max(f(\mathbf{x})) \geq\int_0^\infty\Gamma(r)g_{\mathbf{x}^*}(r)dr    
\end{align}
Similarly, we also have $\max(f(\mathbf{x})) \geq f(\mathbf{x}^*)\geq \min(f(\mathbf{x}))$. Thus if $\min(f(\mathbf{x})) > (1-k)\max(f(\mathbf{x}))$ we will then have:
\begin{align}
    f(\mathbf{x}^*) \geq\min(f(\mathbf{x})) > (1-k)\max(f(\mathbf{x}))\geq(1-k)\int_0^\infty\Gamma(r)g_{\mathbf{x}^*}(r)dr
\end{align}
Thus, we have a sufficient but not necessary conditions where minimum dispersal being a evolutionary stable strategy being:
\begin{align}
\frac{\min(f(\mathbf{x}))}{\max(f(\mathbf{x}))} > 1-k    
\end{align}
This means that if there are little variation across the landscape relative to the cost, dispersal will not evolve without kin selection. 

If we make some additional assumptions, we can generate more insight from this model. Specifically, let's assume the landscape is formed by a Gaussian process with spatial autocorrelation with a constant mean and variance. This means that if we randomly sample a point in the landscape, the productivity comes from a normal distribution with mean $\mu$ and standard deviation $\sigma$, and if we randomly sample two points a distance $r$ apart, they are correlated with correlation coefficient $\rho(r)$. From the assumption of gaussian process, we know that:
\begin{align}
    E(f(\mathbf{x}_2)|f(\mathbf{x}_1)) = \mu + \rho(\|\mathbf{x}_2 - \mathbf{x}_1\|)(\mathbf{x}_1-\mu)
\end{align}
Let's now define $\Omega_{s}$ as all the points $\mathbf{x}$ where $f(\mathbf{x}) = s$, we can then deduce that:
\begin{align}
    E_{\Omega_s}(g_x(r)) = \mu + \rho(r)(s-\mu)
\end{align}
Where $E_{\Omega_s}$ is the expected value over all points with habitat productivity of $s$. Thus:
\begin{align}
    E_{\Omega_s}(\mathbf{i}-\mathbf{o}) = p_{min}(1-k)\int_0^\infty\Gamma(r)(\mu + \rho(r)(s-\mu))dr - sp_{min}
\end{align}
Expanding and rearranging to get:
\begin{align}
    &E_{\Omega_s}(\mathbf{i}-\mathbf{o}) =\\& p_{min}(1-k)(\mu\int_0^\infty\Gamma(r)dr +(s-\mu)\int_0^\infty\Gamma(r)\rho(r)dr)-sp_{min}
\end{align}
Since $\Gamma(r)$ is the distribution of distance traveled of a dispersed individual conditioned on it's survival, it's integral must evaluate to 1 and thus we have:
\begin{align}
    &E_{\Omega_s}(\mathbf{i}-\mathbf{o}) =\\& p_{min}(1-k)(\mu+(s-\mu)\int_0^\infty\Gamma(r)\rho(r)dr)-sp_{min}
\end{align}
If substituting in:
\begin{align}
   I =  \int_0^\infty\Gamma(r)\rho(r)dr
\end{align}
Where $I$ is the expected autocorrelation between the local patch quality and an expected destination of disperser, we have the condition where the expected net flux being greater than or equal to $0$ is:
\begin{align}
    E_{\Omega_s}(\mathbf{i}-\mathbf{o}) \geq0\Leftrightarrow(I(1-k)-1)s \geq (1-k)(I-1)\mu
\end{align}
And since $-1\leq I\leq1$ as it is the expected value of a function ($\rho$) strictly between 1 and -1,  we have $(I(1-k)-1)<0$. For a landscape that is either non-homogeneous ($I < 1$) or with some cost of dispersal ($k>0$), we have:
\begin{align}
    E_{\Omega_s}(\mathbf{i}-\mathbf{o}) \geq0\Leftrightarrow s \leq \frac{(1-k)(I-1)}{I(1-k)-1}\mu
\end{align}
If we define:
\begin{align}
    Q\equiv\frac{(1-k)(I-1)}{I(1-k)-1}
\end{align}
than we have:
\begin{align}
    E_{\Omega_s}(\mathbf{i}-\mathbf{o}) \geq0\Leftrightarrow s \leq Q\mu
\end{align}
Thus, for all $s$ satisfying the inequality, it is guaranteed that at \textit{at least some} of the points with productivity $s$, minimal dispersal propensity is not evolutionarily stable on an infinitely large plain. Or equivalently, at points that satisfies the inequality dispersal is expected to be selected for. Note the inequality is not to say that minimal dispersal is stable where it does not hold, provided the dominant strategy is to disperse minimally everywhere else; rather, it is saying \textit{where it does hold}, we expect minimal dispersal to be unstable in \textit{some} of those points. 

To gain some further understanding, it pays to further analyze the fraction $Q$ and integral $I$. Whilst it is possible to do more analytical work, directly plotting out the fraction provides most intuition. Since we know $k$ lies between 0 and 1 and $I$ between -1 and 1, we can focus our analysis to those cases. By plotting the fraction $Q$ with respect to the expected correlation between the origin and destination of an disperser ($I$) with different values of dispersal cost ($k$), we get a set of concave, monotonically decreasing, curves which at $I = 0$ takes the value of $1-k$, and, regardless of k, always reduce to 0 at $I = 1$. It can thus be seen low level of autocorrelation, or indeed, negative autocorrelation can help the evolution of dispersal. This is intuitively pleasing as strong autocorrelation means that the nearby landscape is more similar to the current patch then by chance. There are thus little to gain by dispersing. On the other hand, if we have negative autocorrelation, just beside a low quality patch is likely a high quality one, thus dispersal can be favoured in these low quality patches.
\begin{figure}
    \centering
\includegraphics[width = 15cm]{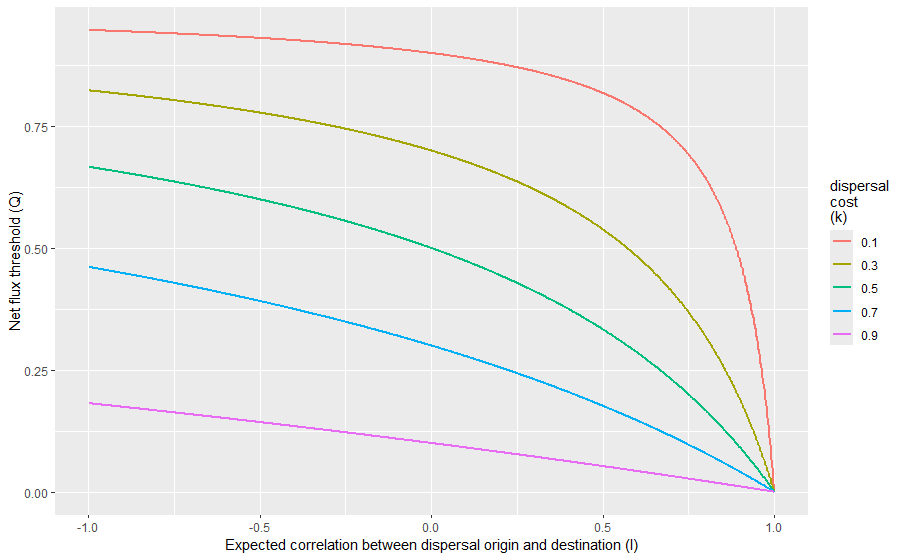}
    \caption{The relation ship between net flux threshold (Q) and expected correlation between qualities of origin and destination of dispersal at different dispersal cost (I). If a patch has lower quality than $Q\mu$ where $\mu$ is the mean quality on the landscape, then we expect the net flux is greater than zero when minimal dispersal across the landscape is the population strategy. This means that we expect adaptive dispersal to evolve at these locations.}
    \label{fig:dispersal_fraction}
\end{figure}

To also get a taste of what $I$ might look like, we may look at the case of an exponential dispersal kernel and a exponential decay in autocorrelation. Note whilst these two functions are chosen due to mathematical convenience, the insight we can gain is much more general and applies to other dispersal kernel and strictly positive autocorrelation function that decays with distance as well: more on that later. Suppose the dispersal kernel have the mean dispersal length of $\beta$, and autocorrelation have the characteristic length $\lambda$, we thus have:
\begin{align}
    \Gamma(r) &= \frac{1}{\beta} e^{-\frac{r}{\beta}}
    \\
    \rho(r)&=e^{-\frac{r}{\lambda}}
\end{align}
and:
\begin{align}
    I = \frac{1}{\beta}\int_0^\infty e^{-\frac{r}{\beta}}e^{-\frac{r}{\lambda}}dr = \frac{1}{\beta}\int_0^\infty e^{-(\frac{\beta + \lambda}{\beta\lambda}){r}}dr = \frac{\lambda}{\beta + \lambda}
\end{align}
Which means that, when dispersal and autocorrelation have similar characteristic length, we have $I\simeq0.5$. If individuals tend to disperse further than the autocorrelative structure of the landscape, then we have $\beta\gg \lambda$ and thus $I\simeq0$.  This matches our biological intuition where the disperser is effectively randomly sampling habitat across the landscape, in which case the average payoff provided it dispersed successfully is $\mu$, meaning it only pays to disperse if the current habitat is worse then $(1-k)\mu$. On the other hand, if the typical length of the landscape structure is much larger than the dispersal length, or $\lambda \gg \beta$ we have $I \simeq \lambda/\lambda = 1$ which the above graph suggest any none-zero cost would lead to dispersal being favoured not even in the worst patches. This again appeals to our intuition as if dispersal distances are much smaller than the landscape structure, then dispersing would not lead to much change to the focal individual. Thus, since it is a costly endeavor, it is selected against in the absent of kin selection. 

To see how the qualitative result is not an artifact of our choice of dispersal kernel and autocorrelation function, but rather a much more general pattern, one can appeal to our own biological intuition, as briefly outlined above. Or alternatively, one can consider the mathematical form of the integral. If we have the length scale of autocorrelation being much larger than the dispersal kernel, then the autocorrelation function would be close to 1 throughout the range in which the dispersal kernel have appreciable density, as such, the integral would be close to one; and conversely, if the dispersal kernel has a much larger scale then the autocorrelative structure, then throughout the range in which the dispersal kernel have appreciable density, the autocorrelation will be close to 0 except but a sliver in the extremely close range, where the overall density from the kernel is rather low. Resulting in the integral being close to 0.

It may also be of interest to look at autocorrelative structure that contains negative correlation, such is common in disturbed landscape \parencite{biswasDisturbanceIncreasesNegative2017}, which is increasingly pervasive. As before, for mathematical convenience a certain functional form is chosen though the qualitative result is general. Here we will use the same dispersal kernel with the autocorrelation function:
\begin{align}
    \rho(r) = e^{-\frac{r}{\lambda}}\cos(\frac{2\pi r}{a})
\end{align}
Where $\lambda$, as before, is the characteristic length scale for decay, and $a$ the length scale negative autocorrelation. For negative autocorrelation, periodic function is chosen since if the correlation between this patch and a patch $r$ away is negative, and so are the correlation between a patch $r$ and another $2r$ away. Thus the correlation between the current patch and ones $2r$ away should be positive. We then have:
\begin{align}
    I = \frac{1}{\beta}\int_0^\infty e^{-\frac{(\beta + \lambda)r}{\beta\lambda}}\cos(\frac{2\pi r}{a})dr = \frac{a^2\lambda(\beta+\lambda)}{a^2(\beta+\lambda)^2 + 4\pi^2\beta^2\lambda^2} 
\end{align}
If we then treat $\beta$ and $\lambda$ as fixed, we can look at the effect of scale of negative correlation on the evolution of dispersal. We have:
\begin{align}
    \lim_{a\rightarrow0}I &= 0\\
    \lim_{a\rightarrow\infty}I &= \frac{\beta\lambda}{\beta + \lambda}\\
    \frac{dI}{da}&=\frac{8a\beta^3\lambda^3(\beta+\lambda)\pi^2}{(a^2(\beta+\lambda)^2 + 4\pi^2\beta^2\lambda^2)^2} > 0
\end{align}
Thus, if the length scale of negative correlation is much smaller then that of dispersal or decay of autocorrelation, then the expected autocorrelation tends to zero. If on the other hand, the length scale of negative correlation is much larger than the scale of dispersal or the scale of autocorrelation, then the expected autocorrelation approaches the case where there are no negative autocorrelation. Further, expected autocorrelation is a increasing function of the length scale of negative autocorrelation, thus for intermediate scale, the value is somewhere in between. This again is intuitively pleasing, with the length scale of negative autocorrelation being very low, the neighbourhood is then essentially filled with both patches of very low and very high quality, as a result, the average quality of destination is effectively $\mu$. On the other hand if it is very high, than it's effect is not felt in the neighbourhood. FOr something in between, since some of the neighbouring patches will be negatively correlated, thus for a patch worse then average it is now more likely to have a higher quality patch around, and thus the evolution of dispersal is favoured. Though these results are in some cases particular to the equation we have chosen, indeed in this case we still have $I\geq0$ which is not necessarily true, the overall message where some negative autocorrelation at an appropriate scale facilitates the evolution of dispersal still holds. 
\section{Existence and Uniqueness of ESS}
\label{Appendix: proof}
A proof for the existence and uniqueness, using above algorithm, is given below. To make the proof easier instead of working with the minimal dispersal probability $p_{min}$ and dispersal probability vector $p$, the proof works with the equivalent minimal dispersal mass $o_{min} = q*p_{min}$, dispersal mass $o = p \odot q$ where $\odot$ denotes the element wise product and $o_{max} = q$

\begin{theorem}
Let $o_{\min}, o_{\max} \in \mathbb{R}^n$ with 
\[
0 < (o_{\min})_i \leq (o_{\max})_i \quad \text{for all } i = 1, \dots, n,
\]
and let $T \in \mathbb{R}^{n \times n}$ be a square transition matrix satisfying:
\begin{enumerate}
    \item $T$ is positive semidefinite and ergodic;
    \item Each column sum of $T$ is less than $1$ (hence the spectral radius of $T$ is $< 1$).
\end{enumerate}
An \emph{ESS} is defined as a vector $o \in \mathbb{R}^n$ such that $o_{\min,i} \leq o_i \leq o_{\max,i}$ for all $i$, and moreover
\[
\begin{cases}
o_i \geq (To)_i & \text{if } o_i = o_{\min,i}, \\[6pt]
o_i = (To)_i & \text{if } o_{\min,i} < o_i < o_{\max,i}, \\[6pt]
o_i \leq (To)_i & \text{if } o_i = o_{\max,i}.
\end{cases}
\]
\end{theorem}

Such ESS exists and is unique.
\subsection*{Roadmap}
\begin{description}
  \item[Step 1: Fixed Point Characterization.]
  Define the update map $f$ and the fixed point set $\mathrm{FP}$, and prove the coordinatewise characterization of $\mathrm{FP}$.

  \item[Step 2: The Admissible Set $S$.]
  Introduce $S$ and show $\{\mathrm{ESS}\}=S\cap \mathrm{FP}$, reducing the problem to finding fixed points of $f$ within $S$.

  \item[Step 3: Monotone Iteration and Convergence.]
  Prove that $S$ is closed, $T$ is isotone, $f$ is extensive and preserves $S$, and that the iteration $a_{n+1}=f(a_n)$ with $a_0=o_{\min}$ is monotone increasing, bounded, and converges to $a_\infty\in S$.

  \item[Step 4: Existence of an ESS.]
  Show $f$ is continuous on the order interval $[o_{\min},o_{\max}]$, pass to the limit in $a_{n+1}=f(a_n)$ to get $f(a_\infty)=a_\infty$, hence $a_\infty\in S\cap \mathrm{FP}$ is an ESS.

  \item[Step 5: Minimality via Monotonicity.]
  Prove $f$ is monotone: $x\le y \Rightarrow f(x)\le f(y)$. Deduce that the limit $a_\infty$ is the minimal ESS (coordinatewise) among all ESS.

  \item[Step 6: Uniqueness of the ESS.]
  For any two ESS $o,o'$ with $o'\ge o$, set $\Delta=o'-o\ge 0$ and show $\Delta\le T\Delta$. Using $\rho(T)<1$, conclude $\Delta=0$, hence $o'=o$. Therefore the ESS is unique.
\end{description}

\subsection*{Step 1: Fixed Point Characterization}

\begin{definition}
Define the function $f : \mathbb{R}^n \to \mathbb{R}^n$ by
\[
\bigl(f(x)\bigr)_i =
\begin{cases}
x_i, & \text{if } x_i \geq (Tx)_i, \\[6pt]
\min\bigl( (o_{\max})_i,\; (Tx)_i \bigr), & \text{if } x_i < (Tx)_i.
\end{cases}
\]
We denote by $\mathrm{FP}$ the set of fixed points of $f$:
\[
\mathrm{FP} := \{ x \in \mathbb{R}^n : f(x) = x \}.
\]
\end{definition}

\begin{lemma}[Equivalent characterization of $\mathrm{FP}$]
\label{lemma:FP}
The set $\mathrm{FP}$ admits the following equivalent description:
\[
\mathrm{FP} = \Bigl\{ x \in \mathbb{R}^n : x_i \geq (Tx)_i \text{ or } x_i = (o_{\max})_i \quad \forall i \Bigr\}.
\]
\end{lemma}

\begin{proof}
If $x$ satisfies the right-hand condition, then $f(x) = x$ by direct substitution, hence $x \in \mathrm{FP}$.  

Conversely, suppose $x \in \mathrm{FP}$ but $x \notin$ the right-hand set. Then there exists some index $j$ with $x_j \neq (o_{\max})_j$ and $x_j < (Tx)_j$. By the definition of $f$, we then have
\[
f(x)_j = \min\bigl( (o_{\max})_j, (Tx)_j \bigr) \neq x_j,
\]
contradicting $f(x) = x$. Therefore the two sets are equal.
\end{proof}

\subsection*{Step 2: The Admissible Set $S$}

\begin{definition}
Define the set
\[
S := \Bigl\{ x \in \mathbb{R}^n : (o_{\min})_i \leq x_i \leq (o_{\max})_i, \; (Tx)_i \geq x_i \;\; \text{whenever } x_i \neq (o_{\min})_i \Bigr\}.
\]
\end{definition}

\begin{proposition}
The set of all ESS coincides with the intersection
\[
\{\text{ESS}\} = S \cap \mathrm{FP}.
\]
\end{proposition}

\begin{proof}
We show both inclusions.

\medskip

($\subseteq$) Suppose $o$ is an ESS. By definition, $o_{\min,i} \leq o_i \leq o_{\max,i}$ for all $i$.  
Now check each case:
\begin{itemize}
    \item If $o_i = o_{\min,i}$, this satisfies the condition of $S$. The ESS condition requires $o_i \geq (To)_i$ satisfying the condition of $\mathrm{FP}$. 
    \item If $o_{\min,i} < o_i < o_{\max,i}$, the ESS condition requires $o_i = (To)_i$. Then in particular $(To)_i \geq o_i$, satisfying the condition for $S$. Moreover, $o_i \geq (To)_i$ also satisfies the condition of $\mathrm{FP}$.
    \item If $o_i = o_{\max,i}$, this automatically satisfies the condition for $\mathrm{FP}$. The ESS condition requires $o_i \leq (To)_i$ satisfying the condition of $S$.
\end{itemize}
Thus $o \in S \cap \mathrm{FP}$.

\medskip
\noindent
($\supseteq$) Suppose $o \in S \cap \mathrm{FP}$. Then by definition of $S$ we have
\[
o_{\min,i} \leq o_i \leq o_{\max,i}, \quad \text{and if } o_i \neq o_{\min,i}, \; (To)_i \geq o_i.
\]
From $\mathrm{FP}$ we also know that for each $i$, either $o_i \geq (To)_i$ or $o_i = o_{\max,i}$.

Now check the three ESS cases:
\begin{itemize}
    \item If $o_i = o_{\min,i}$, then the $\mathrm{FP}$ condition gives $o_i \geq (To)_i$ or $o_i = o_{\max,i}$. Since $o_i = o_{\min,i} \neq o_{\max,i}$, we must have $o_i \geq (To)_i$. Hence $o_i = o_{\min,i}$ implies $o_i \geq (To)_i$, which is exactly the ESS condition.
    \item If $o_{\min,i} < o_i < o_{\max,i}$, then by $S$ we know $(To)_i \geq o_i$, while by $\mathrm{FP}$ we know $o_i \geq (To)_i$. Thus we conclude $o_i = (To)_i$, which is the ESS condition.
    \item If $o_i = o_{\max,i}$, then by $\mathrm{FP}$ we have $o_i \geq (To)_i$, again matching the ESS condition.
\end{itemize}
Therefore $o$ satisfies all defining conditions of an ESS, so $o$ is an ESS.

\medskip
Since both inclusions hold, the sets are equal: $\{\text{ESS}\} = S \cap \mathrm{FP}$.
\end{proof}

\subsection*{Step 3: Monotone Iteration and Convergence}

\begin{lemma}
The set $S$ is closed.
\end{lemma}

\begin{proof}
Recall that a set $\Omega \subset \mathbb{R}^n$ is closed if its complement $\Omega^c$ is open.  

We can decompose the complement of $S$ as
\[
S^c = nS_1 \cup nS_2 \cup nS_3,
\]
where
\begin{align*}
nS_1 &:= \{ x \in \mathbb{R}^n : x_i < (o_{\min})_i \text{ for some } i \}, \\
nS_2 &:= \{ x \in \mathbb{R}^n : x_i > (o_{\max})_i \text{ for some } i \}, \\
nS_3 &:= \{ x \in \mathbb{R}^n : (Tx)_i < x_i \text{ for some } i \}.
\end{align*}

\noindent
It suffices to show each $nS_k$ is open:
\begin{itemize}
    \item For $nS_1$: if $x^* \in nS_1$, then there exists $j$ with $x^*_j < (o_{\min})_j$. Let $\delta = (o_{\min})_j - x^*_j > 0$. For any $x'$ with $\|x' - x^*\| < \delta$, we have $x'_j < (o_{\min})_j$, so $x' \in nS_1$. Thus $nS_1$ is open.
    \item For $nS_2$: the same reasoning applies with $x^*_j > (o_{\max})_j$.
    \item For $nS_3$: note that each map $x \mapsto (Tx)_i - x_i$ is continuous as a linear function of $x$. The set
    \[
    \{ x : (Tx)_i - x_i < 0 \}
    \]
    is an open half-space. Taking the union over all indices $i$, we obtain
    \[
    nS_3 = \bigcup_{i=1}^n \{ x : (Tx)_i - x_i < 0 \},
    \]
    which is open as a finite union of open sets.
\end{itemize}
Therefore $S^c$ is open, hence $S$ is closed.
\end{proof}

\begin{lemma}[Isotonicity of the transition map]
\label{lem:isotone}
 
If $x,y\in\mathbb{R}^n$ satisfy $x\le y$ coordinatewise, then $Tx\le Ty$ coordinatewise.
\end{lemma}

\begin{proof}
Write $(Tx)_i=\sum_{j}T_{ij}x_j$ and $(Ty)_i=\sum_{j}T_{ij}y_j$. Since $T_{ij}\ge 0$ and $x_j\le y_j$ for each $j$,
\[
(Tx)_i=\sum_{j}T_{ij}x_j \;\le\; \sum_{j}T_{ij}y_j=(Ty)_i \quad \text{for all } i.
\]
\end{proof}

\begin{lemma}
\label{lem:f-extensive}
If $x\in S$, then $f(x)\ge x$ coordinatewise.
\end{lemma}

\begin{proof}
Fix $i\in\{1,\dots,n\}$. There are two cases by the definition of $f$.
\begin{itemize}
    \item If $x_i\ge (Tx)_i$, then $f(x)_i=x_i\ge x_i$.
    \item If $x_i<(Tx)_i$, then $f(x)_i=\min\{(o_{\max})_i,(Tx)_i\}$. Since $x\in S$ implies $(o_{\min})_i\le x_i\le (o_{\max})_i$, we have $(o_{\max})_i\ge x_i$. Also $x_i<(Tx)_i$ in this case. Hence $f(x)_i\ge x_i$.
\end{itemize}
Thus $f(x)\ge x$.
\end{proof}

\begin{proposition}[$f$ preserves $S$]
\label{prop:f-preserves-S}
If $x\in S$, then $f(x)\in S$.
\end{proposition}

\begin{proof}
We verify the two defining properties of $S$.

\smallskip
\emph{Bounds.}
By Lemma~\ref{lem:f-extensive}, $f(x)_i\ge x_i\ge (o_{\min})_i$. Also by definition of $f$, 
$f(x)_i\le \min\{(o_{\max})_i,(Tx)_i\}\le (o_{\max})_i$. Hence 
\((o_{\min})_i\le f(x)_i\le (o_{\max})_i\) for all $i$.

\smallskip

Let $i$ be such that $f(x)_i\neq (o_{\min})_i$. We must show $(Tf(x))_i\ge f(x)_i$.

There are two cases.

\underline{Case 1:} $f(x)_i=x_i$ (which occurs exactly when $x_i\ge (Tx)_i$).  
If additionally $x_i\neq (o_{\min})_i$, then since $x\in S$ we have $(Tx)_i\ge x_i$. 
 By Lemma~\ref{lem:f-extensive}, $f(x)\ge x$;
hence Lemma~\ref{lem:isotone} gives $Tf(x)\ge Tx$, so
\[
(Tf(x))_i \;\ge\; (Tx)_i \;=\; x_i \;=\; f(x)_i.
\]
If instead $x_i=(o_{\min})_i$, then $f(x)_i=x_i=(o_{\min})_i$, contradicting the premise $f(x)_i\neq (o_{\min})_i$. 
Thus this subcase cannot occur.

\underline{Case 2:} $f(x)_i=\min\{(o_{\max})_i,(Tx)_i\}$ (which occurs when $x_i<(Tx)_i$).  
Lemma~\ref{lem:f-extensive} gives $f(x)\ge x$, so Lemma~\ref{lem:isotone} yields $Tf(x)\ge Tx$. Therefore
\[
(Tf(x))_i \;\ge\; (Tx)_i \;\ge\; \min\{(o_{\max})_i,(Tx)_i\} \;=\; f(x)_i.
\]

In both cases, $(Tf(x))_i\ge f(x)_i$ whenever $f(x)_i\neq (o_{\min})_i$. Hence $f(x)\in S$.
\end{proof}

\begin{lemma}
\label{lem:omin-in-S}
$o_{\min}\in S$.
\end{lemma}

\begin{proof}
Clearly $(o_{\min})_i\le (o_{\max})_i$ for all $i$, so the bounds hold. 
For the one-sided condition in the definition of $S$, note that there is no index with $x_i\neq (o_{\min})_i$ when $x=o_{\min}$; hence the requirement is vacuous.
\end{proof}

\begin{proposition}[Monotone bounded iteration and convergence in $S$]
\label{prop:an-converges-in-S}
Define $a_0=o_{\min}$ and $a_{n+1}=f(a_n)$ for $n\ge 0$. Then:
\begin{enumerate}
    \item $a_n\in S$ for all $n$;
    \item $a_{n+1}\ge a_n$ coordinatewise for all $n$;
    \item $a_n\le o_{\max}$ coordinatewise for all $n$;
    \item $a_n$ converges coordinatewise to a limit $a_\infty\in S$.
\end{enumerate}
\end{proposition}

\begin{proof}
(1) By Lemma~\ref{lem:omin-in-S}, $a_0\in S$. If $a_n\in S$, then by Proposition~\ref{prop:f-preserves-S}, $a_{n+1}=f(a_n)\in S$. Induction yields $a_n\in S$ for all $n$.

(2) By Lemma~\ref{lem:f-extensive}, $a_{n+1}=f(a_n)\ge a_n$.

(3) From the definition of $f$, for each $i$,
\(
a_{n+1,i}=f(a_n)_i\le (o_{\max})_i
\)
and trivially $a_{0,i}=(o_{\min})_i\le (o_{\max})_i$. Hence $a_n\le o_{\max}$ for all $n$.

(4) Each coordinate sequence $\{a_{n,i}\}_{n\ge 0}$ is monotone increasing and bounded above by $(o_{\max})_i$, hence convergent in $\mathbb{R}$. Let $a_{\infty,i}=\lim_{n\to\infty}a_{n,i}$ and $a_\infty=(a_{\infty,i})_{i=1}^n$. Since $S$ is closed (proved earlier), the limit $a_\infty\in S$.
\end{proof}

\subsection*{Step 4: Existence of an ESS}
\begin{lemma}[Continuity of $f$ on the order interval]
\label{lem:f-continuous-box}
Let
\[
\mathcal{B}\;:=\;\{\,x\in\mathbb{R}^n:\ o_{\min}\le x\le o_{\max}\ \text{coordinatewise}\,\}.
\]
On $\mathcal{B}$ we have, for each $i$,
\[
f(x)_i \;=\; \min\!\bigl\{(o_{\max})_i,\; \max\{x_i,(Tx)_i\}\bigr\}.
\]
Consequently $f$ is continuous on $\mathcal{B}$, being a composition of continuous maps (the identity, the linear map $x\mapsto Tx$, and coordinatewise min/max with a fixed value). Since $S\subseteq \mathcal{B}$ trivially, $f$ is continuous on $S$ as well.
\end{lemma}

\begin{proof}
Fix $x\in\mathcal{B}$ and $i\in\{1,\dots,n\}$. Because $x\le o_{\max}$ on $\mathcal{B}$, we have
\[
\min\!\bigl\{(o_{\max})_i,\; \max\{x_i,(Tx)_i\}\bigr\}
=
\begin{cases}
x_i, & \text{if } x_i\ge (Tx)_i,\\[4pt]
\min\!\bigl\{(o_{\max})_i,(Tx)_i\bigr\}, & \text{if } x_i<(Tx)_i,
\end{cases}
\]
which matches the definition of $f(x)_i$ on $\mathcal{B}$. Hence the displayed identity holds.

For continuity, note that:
(i) $x\mapsto x$ (identity) is continuous; 
(ii) $x\mapsto Tx$ is continuous as $T$ is linear; 
(iii) $(u,v)\mapsto \max\{u,v\}$ and $(u,v)\mapsto \min\{u,v\}$ are continuous on $\mathbb{R}^2$ (e.g.\ $\max\{u,v\}=\tfrac12(u+v+|u-v|)$ and $\min\{u,v\}=\tfrac12(u+v-|u-v|)$).
Therefore, each coordinate map
\[
x\ \longmapsto\ \min\!\bigl\{(o_{\max})_i,\; \max\{x_i,(Tx)_i\}\bigr\}
\]
is continuous on $\mathcal{B}$, and hence $f$ is continuous on $\mathcal{B}$. Since $S\subseteq\mathcal{B}$, the restriction $f\restriction_S$ is continuous.
\end{proof}

\begin{proposition}[Limit point is a fixed point: order argument]
\label{prop:limit-is-FP-order}
Let $a_0=o_{\min}$ and $a_{n+1}=f(a_n)$, and let $a_\infty=\lim_{n\to\infty}a_n$ as in Proposition~\ref{prop:an-converges-in-S}. Then $a_\infty\in\mathrm{FP}$.
\end{proposition}

\begin{proof}
By Proposition~\ref{prop:an-converges-in-S}, $(a_n)$ is coordinatewise nondecreasing, contained in $S$, and converges to $a_\infty\in S$. In particular,
\begin{equation}
\label{eq:below-limit}
a_n \le a_\infty \quad \text{coordinatewise for all } n.
\end{equation}
Assume, for a contradiction, that $f(a_\infty)\neq a_\infty$. Since $a_\infty\in S$, Lemma~\ref{lem:f-extensive} gives $f(a_\infty)\ge a_\infty$, so there exists an index $i$ with
\[
f(a_\infty)_i \;>\; a_{\infty,i}.
\]
Set $\gamma := f(a_\infty)_i - a_{\infty,i} > 0$. By Lemma~\ref{lem:f-continuous-box} and $a_n\to a_\infty$ in $S$, there exists $N$ such that for all $n\ge N$,
\[
\bigl|f(a_n)_i - f(a_\infty)_i\bigr| < \gamma/2.
\]
Hence, for $n\ge N$,
\[
a_{n+1,i}\;=\; f(a_n)_i \;>\; f(a_\infty)_i - \gamma/2 \;=\; a_{\infty,i} + \gamma/2 \;>\; a_{\infty,i},
\]
which contradicts \eqref{eq:below-limit} (applied to $a_{n+1}$). Therefore $f(a_\infty)=a_\infty$, i.e.\ $a_\infty\in\mathrm{FP}$.
\end{proof}

\begin{corollary}[Existence of an ESS]
\label{cor:ESS-exists-order}
Since $a_\infty\in S$ (Proposition~\ref{prop:an-converges-in-S}) and $a_\infty\in\mathrm{FP}$ (Proposition~\ref{prop:limit-is-FP-order}), the characterization 
$\{\text{ESS}\}=S\cap\mathrm{FP}$ implies that $a_\infty$ is an ESS.
\end{corollary}

We have thus proved the existence of ESS and provided an algorithm to find it. 

\subsection*{Step 5: Minimality via Monotonicity}
\begin{lemma}[Monotonicity of $f$]
\label{lem:f-monotone}
If $x,y\in\mathcal{B}$ satisfy $x\le y$ coordinatewise, then $f(x)\le f(y)$ coordinatewise.
\end{lemma}

\begin{proof}
Fix $i\in\{1,\dots,n\}$. By Lemma~\ref{lem:isotone}, $x\le y$ implies $(Tx)_i\le (Ty)_i$. Hence
\[
\max\{x_i,(Tx)_i\}\;\le\;\max\{y_i,(Ty)_i\}.
\]
Taking the minimum with $(o_{\max})_i$ preserves the inequality, so
\[
f(x)_i \;=\; \min\{(o_{\max})_i, \max\{x_i,(Tx)_i\}\}
       \;\le\; \min\{(o_{\max})_i, \max\{y_i,(Ty)_i\}\}
       \;=\; f(y)_i.
\]
Thus $f(x)\le f(y)$ coordinatewise.
\end{proof}

\begin{corollary}
\label{cor:order-preserve}
If $o$ is an ESS and $x\in\mathcal{B}$ with $x\le o$, then $f(x)\le f(o)=o$.
\end{corollary}

\begin{proof}
Immediate from Lemma~\ref{lem:f-monotone} and the fixed point property $f(o)=o$.
\end{proof}

\begin{proposition}[Existence of a minimal ESS]
\label{prop:minimal-ESS}
There exists an ESS $o^\ast$ such that $o^\ast\le o$ for all $o\in\{\mathrm{ESS}\}$. Moreover, $o^\ast$ is realized as the limit $a_\infty$ of the monotone iteration $a_{n+1}=f(a_n)$ starting from $a_0=o_{\min}$.
\end{proposition}

\begin{proof}
Suppose, for contradiction, that no such $o^\ast$ exists. Then define a vector $v\in\mathbb{R}^n$ by
\[
v_i := \min\{\,o_i : o\in\{\mathrm{ESS}\}\,\}.
\]
By construction $v\le o$ for all $o\in\{\mathrm{ESS}\}$. Since $a_0=o_{\min}$ and $o_{\min}\le o$ for all ESS by definition, we also have $a_0\le v$. By Corollary~\ref{cor:order-preserve}, this inequality is preserved under iteration, hence
\[
a_n \le v \qquad\text{for all } n.
\]
Therefore the limit $a_\infty=\lim_{n\to\infty}a_n$ satisfies $a_\infty\le v$.

On the other hand, if $v$ is not itself an ESS, then for each $o\in\{\mathrm{ESS}\}$ there exists at least one coordinate $i$ with $o_i-v_i>0$. This means $a_\infty$ cannot equal any $o\in\{\mathrm{ESS}\}$, because $a_\infty\le v < o$ in that coordinate. But by Corollary~\ref{cor:ESS-exists-order}, $a_\infty$ is an ESS. This contradiction shows that $v$ must itself belong to $\{\mathrm{ESS}\}$.

Thus $o^\ast:=v$ is an ESS with $o^\ast\le o$ for all $o\in\{\mathrm{ESS}\}$. By the iterative construction, $o^\ast=a_\infty$.
\end{proof}
\subsection*{Step 6: Uniqueness of the ESS}
\begin{proposition}[No two comparable ESS]
\label{prop:no-two-comparable-ESS}
Let $o,o'\in\{\mathrm{ESS}\}=S\cap\mathrm{FP}$. If $o'\ge o$ coordinatewise, then $o'=o$.
\end{proposition}

\begin{proof}
Assume $o,o'\in S\cap\mathrm{FP}$ and $o'\ge o$. Set $\Delta:=o'-o\ge 0$. We will show $\Delta=0$.

\medskip\noindent
\emph{Step 1: A pointwise inequality $\Delta\le T\Delta$.}
Fix $i\in\{1,\dots,n\}$. We distinguish by $o_i$:

\smallskip
\underline{(a) $o_i=o_{\max,i}$.} Then $o'_i\ge o_i=o_{\max,i}$ forces $o'_i=o_{\max,i}$, so $\Delta_i=0\le (T\Delta)_i$.

\smallskip
\underline{(b) $o_{\min,i}<o_i<o_{\max,i}$.} Since $o\in S\cap\mathrm{FP}$ and $o_i<o_{\max,i}$, the characterization of ESS yields $(To)_i=o_i$. For $o'$, either $o'_i<o_{\max,i}$ (then $(To')_i=o'_i$) or $o'_i=o_{\max,i}$ (then $(To')_i\ge o'_i$). In both subcases,
\[
(T\Delta)_i \;=\; (To')_i-(To)_i \;\ge\; o'_i-o_i \;=\; \Delta_i.
\]

\smallskip
\underline{(c) $o_i=o_{\min,i}$.} Since $o\in S$, the one-sided condition gives $(To)_i\le o_i$. For $o'$, either $o'_i<o_{\max,i}$ (then $(To')_i=o'_i$) or $o'_i=o_{\max,i}$ (then $(To')_i\ge o'_i$). In both subcases,
\[
(T\Delta)_i \;=\; (To')_i-(To)_i \;\ge\; o'_i-o_i \;=\; \Delta_i.
\]

\smallskip
Thus for all $i$ we have $\Delta_i\le (T\Delta)_i$, i.e.\ $\Delta\le T\Delta$ coordinatewise.

\medskip\noindent
\emph{Step 2: Summing forces equality, hence $T\Delta=\Delta$.}
Let $c_j:=\sum_{i=1}^n T_{ij}$ be the $j$th column sum, with $0 < c_j < 1$ by assumption. Summing the inequality from Step~1 over $i$ gives
\[
\sum_{i=1}^n \Delta_i
\;\le\;
\sum_{i=1}^n (T\Delta)_i
\;=\;
\sum_{i=1}^n \sum_{j=1}^n T_{ij}\Delta_j
\;=\;
\sum_{j=1}^n \!\Bigl(\sum_{i=1}^n T_{ij}\Bigr)\Delta_j
\;=\;
\sum_{j=1}^n c_j\,\Delta_j
\;\le\;
\sum_{j=1}^n \Delta_j.
\]
Where the last equality only holds if $\Delta = 0$ i.e.\ $o'=o$.

\qedhere

\end{proof}

\end{document}